\documentclass[submission]{jsfds}

\RequirePackage[T1]{fontenc}
\usepackage{amsmath,amsthm,amssymb}
\usepackage{amsfonts,dsfont}
\usepackage{pstricks}
\usepackage{comment}
\usepackage{natbib}
\usepackage{colordvi}
\usepackage{graphicx}
\usepackage{rotating}
\arxiv{1012.4078} 

\startlocaldefs
\usepackage{amsthm}
\newtheorem{theorem}{Theorem}[section]
\newtheorem{proposition}[theorem]{Proposition}
\newtheorem{corollary}[theorem]{Corollary}
\newtheorem{lemma}[theorem]{Lemma}

\newtheorem{algorithm}[theorem]{Algorithm}
\newtheorem{definition}[theorem]{Definition}
\theoremstyle{definition}
\newtheorem{example}[theorem]{Example}
\newtheorem{remark}[theorem]{Remark}
\newcommand{\R}{\mathbb{R}} 
\newcommand{\1}{1
} 
\newcommand{\Proba}{\mathbb{P}}
\renewcommand{\P}{\mathbb{P}}
\newcommand{\E}{\mathbb{E}} 
\newcommand{\telque}{:}
\newcommand{\cond}{\: | \:}
\newcommand{\FDR}{\mbox{FDR}}

\newcommand{\FWER}{\mbox{FWER}}
\newcommand{\kFWER}{\mbox{$k$-FWER}}

\newcommand{\mtc}{\mathcal}
\newcommand{\mbf}{\mathbf}

\newcommand{\wt}[1]{{\widetilde{#1}}}
\newcommand{\wh}[1]{{\widehat{#1}}}
\newcommand{\ol}[1]{\overline{#1}}

\newcommand{\ind}[1]{{\mbf{1}\{#1\}}}

\newcommand{\cH}{{\mtc{H}}}
\newcommand{\cX}{{\mtc{X}}}

\newcommand{\cC}{{\mtc{C}}}
\newcommand{\cP}{{\mtc{P}}}
\newcommand{\G}{\widehat{\mathbb{G}}}
 
\usepackage{graphics}
\newcommand{\FDP}{\mbox{FDP}}

\renewcommand{\l}{\ell}
\newcommand{\ER}{\mathcal{E}}

\def\1{1\kern-.20em {\rm l}}
\endlocaldefs

\setmainlanguageenglish
\firstpage{1} 

\begin{document}

\begin{frontmatter}


\title{Type I error rate control for testing many hypotheses: a survey with proofs}

\runtitle{Type I error rate control in multiple testing}
 \alttitle{Une revue du contr\^ole de l'erreur de type I en test multiple}

\begin{aug}
    \author{%
      \fnms{Roquain} \snm{Etienne}%
      \thanksref{t1}%
      \contact[label=e1]{etienne.roquain@upmc.fr}%
    }%
 \affiliation[t1]{UPMC University of Paris 6, LPMA.\\ 
      \printcontact{e1}
    }

\runauthor{Etienne Roquain}
  \end{aug}

\begin{abstract}

This paper presents a survey on some recent advances for the type I error rate control in multiple testing methodology. 
We consider the problem of controlling the $k$-family-wise error rate (kFWER, probability to make $k$ false discoveries or more) and the false discovery proportion (FDP, proportion of false discoveries among the discoveries).  The FDP is controlled either via its expectation, which is the so-called false discovery rate (FDR), or via its upper-tail distribution function.
 We aim at deriving general and unified  results together with concise and simple mathematical proofs.
Furthermore, while this paper is mainly meant to be a survey paper, some new contributions for controlling the kFWER and the upper-tail distribution function of the FDP are provided. In particular, we derive a new procedure based on the quantiles of the binomial distribution that controls the FDP under independence.
 \end{abstract}
 
  \begin{AMSclass}
\kwd{62J15}
\kwd{62G10}
  \end{AMSclass}

\begin{keywords}
\kwd{multiple testing}\kwd{type I error rate}\kwd{false discovery proportion}\kwd{family-wise error}\kwd{step-up}\kwd{step-down}\kwd{positive dependence}
\end{keywords}

\tableofcontents
 
\end{frontmatter}

\section{Introduction}

The problem of testing several null hypotheses has a long history in the statistics literature. With the high-resolution techniques introduced in the recent years, it has known a renewed attention  in many application fields where 
one aims 
to find significant features among several thousands (or millions) of candidates. Classical examples are microarray analysis \cite{WW2007,DL2008,Efron2008,Efron2009}, neuro-imaging analysis \cite{BH2007,PNBL2005} and source detection \cite{Astro2001}. 
For illustration, we detail below the case of microarray data analysis.

\subsection{Multiple testing in microarray data}

In a typical microarray experiment, the level expressions of a set of genes are measured under two different experimental conditions and we aim at finding the genes that are differentially expressed between the two conditions. For instance, when the genes come from tumor cells in the first experimental condition,  while they come from healthy cells in the second, the differentially expressed genes may be involved in the development of this tumor and thus are genes of special interest.
Several techniques exist to perform a statistical test for a single gene, e.g. based on a distributional assumption or on permutations between the two group labels. However, the number of genes $m$ can be large (for instance several thousands), so that non-differentially expressed genes can have a high score of significance by chance. 
In that context, applying the naive, non-corrected procedure (level $\alpha$ for each gene) is unsuitable because it is likely to select (or ``discover") a lot of non-differentially expressed genes (usually called ``false discoveries''). 
For instance, if the $m=10,000$ genes are not differentially expressed (no signal) and $\alpha=0.1$, the non-corrected procedure  makes on average $m\alpha=1,000$ discoveries which are all false discoveries. 
In a more favorable situation where there are only $m_0=5,000$ non-differentially expressed genes among the $m=10,000$ initial genes ($50\%$ of signal), 
 the non-corrected procedure selects some genes, say $r$ genes, for which the expected number of errors is $m_0\alpha=500$. 
 Since the number of discoveries $r$ is not designed to be much larger than the number of false discoveries $m_0\alpha$, the final list of discovered genes is likely to contain an unacceptable part of errors.
 A multiple testing procedure
aims at correcting \textit{a priori} the level of the single tests in order to obtain a list of selected genes for which the ``quantity" of false discoveries is below a nominal level $\alpha$. The ``quantity" of false discoveries is measured by using \textit{global type I error rates}, as for instance the probability to make at least $k$ errors among the discoveries 
($k$-family-wise error rate, $\kFWER$) or the expected proportion of errors among the discoveries (false discovery rate, FDR). 
Finding procedures that control type I error rates is challenging and is what we called here the ``multiple testing issue". 
Furthermore, a feature that increases the complexity of this issue is the presence of dependencies between the single tests. 

Note that the multiple testing issue can be met in microarray analysis under other forms, as for instance when we search co-expressed genes or  genes associated with clinical covariates or outcomes, see Section~1.2 of \cite{DL2008}.

\subsection{Examples of multiple testing settings}\label{sec:exTM}

\begin{example}[Two-sample multiple $t$-tests]
The problem of finding differentially expressed genes in the above microarray example can be formalized as a particular case of a general two-sample multiple testing problem.
Let us observe a couple of two independent samples $$X=(X^1,...,X^n)=\big(Y^1,...,Y^{n_1},Z^1,...,Z^{n_2}\big)\in\R^{m \times n},$$  where $(Y^1,...,Y^{n_1})$ is a family of $n_1$ i.i.d. copies of a random vector $Y$ in $\R^m$ and $(Z^{1},...,Z^{n_2})$ is a family of $n_2$ i.i.d. copies of a random vector $Z$ in $\R^m$ (with $n_1+n_2=n$). In the context of microarray data, $Y^j_i$ (resp. $Z^j_i$), $1\leq i \leq m$, corresponds to the expression level measure of the $i$-th gene for the $j$-th individual of the first (resp. second) experimental condition. 
Typically, the sample size is much smaller than the number of tests, that is, $n\ll m$.
Let the distribution $P$ of the observation $X$ belong to a statistical model given by a 
distribution set $\cP$. 
Assume that $\cP$ is such that $X$ is an integrable random vector and let  $\mu_{i,1}(P)=\E Y_i$ and $\mu_{i,2}(P)=\E Z_i$, for any $i\in\{1,...,m\}$. 
The aim is to decide for all $i$ whether $P$ belongs to the set $\Theta_{0,i}=\{P\in\cP\telque \mu_{i,1}(P)=\mu_{i,2}(P)\}$ or not, that is, we aim at testing the hypothesis
\begin{center}
$H_{0,i}:$ ``$\mu_{i,1}(P)=\mu_{i,2}(P)$" against $H_{1,i}:$ $``\mu_{i,1}(P)\neq \mu_{i,2}(P)"$,
\end{center}
 simultaneously for all $i\in\{1,...,m\}$. 
 Given $P$, the null hypothesis $H_{0,i}$ (sometimes called  the ``null" for short) is said to be true (for $P$) if $P\in \Theta_{0,i}$, that is, if $P$ satisfies  $H_{0,i}$. It is said false (for $P$) otherwise. The index set corresponding to true nulls is denoted by $\cH_0(P)=\{1\leq i\leq m\telque  \mu_{i,1}(P)= \mu_{i,2}(P)\}$. Its complement in $\cH=\{1,...,m\}$ is denoted by $\cH_1(P)$. In the microarray context, $\cH_1(P)=\{1\leq i\leq m\telque  \mu_{i,1}(P)\neq \mu_{i,2}(P)\}$ is thus the index set corresponding to differentially expressed genes. The aim of a multiple testing procedure is thus to recover the (unobservable) set $\cH_1(P)$ given the observation $X$.
 A multiple testing procedure is commonly based on individual test statistics, by rejecting the null hypotheses with a ``large" test statistic. 
Here, the individual test statistic can be the (two-sided) two-sample t-statistic $S_i(X)\propto |\ol{Y}_i-\ol{Z}_i|$, rescaled by the so-called ``pooled" standard deviation.
To provide a uniform normalization for all tests, it is convenient to transform the $S_i(X)$ into the \textit{$p$-value} 
\begin{equation}\label{pvalueform}
p_i(X)=\sup_{P\in\Theta_{0,i}} T_{P,i}(S_i(X)),
\end{equation}
where $T_{P,i}(s)=\P_{X\sim P}(S_i(X)\geq s)$ is the upper-tail distribution function of $S_i(X)$ for $X\sim P\in\Theta_{0,i}$.
Classically, assuming that $Y_i$ and $Z_i$ are Gaussian variables with the same variance, we have for any $P\in\Theta_{0,i}$, $T_{P,i}(s)=2\P(Z\geq s)$, where $Z$ follows a Student distribution with $n-2$ degrees of freedom. In that case, each $p$-value $p_i(X)$  has the property to be uniformly distributed on $(0,1)$ when the corresponding null hypothesis $H_{0,i}$ is true.
Without making this Gaussian assumption, $p$-values can still be built, as we discuss in Remark~\ref{rem:pvalues} below.
Let us finally note that since the $T_{P,i}$ are decreasing, a multiple testing procedure should reject nulls with a ``small" $p$-value. \end{example}

\begin{example}[One-sided testing on the mean of a Gaussian vector]\label{sec:vectorpos}
To give a further illustrating example, we consider the very convenient mathematical framework for multiple testing where we observe a Gaussian vector $X=(X_i)_{1\leq i \leq m} \sim P$, having an unknown mean $\mu(P)=(\mu_i(P))_{1\leq i \leq m}\in\R^m$ and a $m\times m$ covariance  matrix $\Sigma(P)$ with diagonal entries equal to $1$. Let us consider the problem of testing
\begin{center}
$H_{0,i}:$ ``$\mu_{i}(P)\leq 0$" against $H_{1,i}:$ $``\mu_{i}(P)>0"$,
\end{center}
simultaneously for all $i\in\{1,...,m\}$. We can define the $p$-values $p_i=\ol{\Phi}(X_i)$, where $\ol{\Phi}(x)=\P(Z\geq x)$ for $Z\sim\mathcal{N}(0,1)$. Any $p$-value satisfies the following stochastic domination under the null: if $\mu_{i}(P)\leq 0$, we have for all $u\in[0,1]$,
$$
\P(p_i(X)\leq u) \leq \P (\ol{\Phi}(X_i-\mu_i(P))\leq u)=u.
$$
Additionally, more or less restrictive assumptions on $\Sigma(P)$ can be considered to model
different types of dependency of the corresponding $p$-values.
For instance, we can assume that $\Sigma(P)$ has only non-negative entries, that the non-diagonal entries of $\Sigma(P)$ are equal (equi-correlation) or that $\Sigma(P)$ is diagonal.
Finally, the value of the alternative means can be used for modeling  the ``strength of the signal". For instance, to model that the sample size available for each test is $n$, we can set $\mu_i(P)=\tau \sqrt{n}$ for each $\mu_{i}(P)>0$,  where $\tau>0$ is some additional parameter.
\end{example}

\begin{remark}[General construction of $p$-values]\label{rem:pvalues}
In broad generality, when testing the nulls $\Theta_{0,i}$ by rejecting for ``large" values of a test statistic $S_i(X)$, we can always define the associated $p$-values by using \eqref{pvalueform}.
 It is well known that these $p$-values 
 are always stochastically lower-bounded by a uniform variable under the null, that is, $\forall i \in\cH_0(P)$, $\forall u\in[0,1], \: \P(p_i(X)\leq u)\leq u$. This property always holds, even when $S_i(X)$ has a discrete distribution. 
For completeness, we provide this result with a proof in Appendix~\ref{sec:statpvalue}. 
However, the calculation of the $p$-values \eqref{pvalueform} is not always possible, because it requires the knowledge of the distribution of the test statistics under the null, which often relies on strong distributional assumptions on the data. 
Fortunately, in some situations, the $p$-values \eqref{pvalueform} can be approximated by using a randomization technique. The
resulting $p$-values can be shown to enjoy the same stochastic dominance as above (see, e.g., \cite{RW2005} for a recent reference).
For instance, in the two-sample testing problem, permutations of the group labels can be used, which corresponds to use permutation tests (the latter can be traced back to Fisher \cite{Fish1935}).
\end{remark}

\subsection{General multiple testing setting}\label{sec:framework}

In this section, we provide the abstract framework in which multiple testing theory can be investigated in broad generality. 

Let us consider a statistical model, defined by a measurable space $(\cX,\mathfrak{X})$ endowed with a subset $\mathcal{P}$ of distributions on $(\cX,\mathfrak{X})$. Let $X$ denote the observation of the model, with distribution 
$P\in \mathcal{P}$. 
Consider a family $(\Theta_{0,i})_{1\leq i \leq m}$ of $m\geq 2$ subsets  of $\mathcal{P}$. Based on  $X$, we aim at testing the null hypotheses $H_{0,i} : ``P \in \Theta_{0,i}"$ against the alternative  $H_{1,i} : ``P \in \Theta_{0,i}^c"$ simultaneously for all $i \in\{1,...,m\}$. For any $P\in \mathcal{P}$, let $\cH_0(P)=\{1\leq i \leq m\telque P\in\Theta_{0,i}\}$ be the set of the indexes $i$ for which $P$ satisfies $H_{0,i}$, that is, the indexes corresponding to true null hypotheses. 
Its cardinality $|\cH_0(P)|$ is denoted by $m_0(P)$.
Similarly, the set $\{1,...,m\}$  is sometimes denoted by $\cH$.
The set of the false null hypotheses is denoted by $\cH_1(P)=\cH\backslash\cH_0(P)$.
The goal is to recover the set $\cH_1(P)$ based on $X$, that is, to find the null hypotheses that are true/false based on the knowledge of $X$. 
Obviously, the distribution $P$ of $X$ is unknown, and thus so is $\cH_1(P)$. 

The standard multiple testing setting includes the knowledge of 
$p$-values $(p_i(X))_{1\leq i \leq m}$ satisfying 
\begin{equation}\label{equ_defpvalue} \forall P\in\mathcal{P},  \forall i\in\cH_0(P),\:\:\: \forall u\in[0,1], \: \P(p_i(X)\leq u)\leq u.\end{equation}
As a consequence, for each $i\in\{1,...,m\}$, rejecting $H_{0,i}$ whenever $p_i(X)\leq \alpha$ defines a test of level $\alpha$.
As we have discussed in the previous section, property \eqref{equ_defpvalue} can be fulfilled in many situations.
Also, 
in some cases, \eqref{equ_defpvalue} holds with equality, that is,  the $p_i(X)$ are exactly distributed like a uniform variable in $(0,1)$ when $H_{0,i}$ is true. 

\subsection{Multiple testing procedures}

In the remainder of the paper, we use the observation $X$ only through the $p$-value family $\mbf{p}(X)=\{p_i(X),1\leq i \leq m\}$. Therefore, for short, we often drop the dependence in $X$ in the notation and define all quantities as functions of $\mbf{p}=\{p_i,1\leq i \leq m\} \in[0,1]^m$. However, one should keep in mind that the underlying distribution $P$ (the distribution of interest on which the tests are performed) is the distribution of $X$ and not the one of $\mbf{p}$.

A \textit{multiple testing procedure} is defined as a set-valued function $$R:q=(q_i)_{1\leq i \leq m}\in [0,1]^m \longmapsto R(q)\subset \{1,...,m\},$$ taking as input an element of $[0,1]^m$ and returning a subset of  $\{1,...,m\}$. For such a general procedure $R$, we add the technical assumption that  for each $ i\in\{ 1,..., m\}$, the mapping $x\in\mathcal{X}\mapsto\ind{i\in R(\mbf{p}(x))}$ is measurable. The indexes selected by $R(\mbf{p})$ correspond to the rejected null hypotheses, that is, $i\in R(\mbf{p}) \Leftrightarrow ``H_{0,i}$ is rejected by the procedure $R(\mbf{p})$". 
Thus, for each $p$-value family $\mbf{p}$, there are $2^m$ possible outcomes for $R(\mbf{p})$. 
Nevertheless, according to the  stochastic dominance property \eqref{equ_defpvalue} of the $p$-values, a natural rejection region for each $H_{0,i}$ is of the form  $p_i\leq t_i$, for some $t_i\in[0,1]$. 
In this paper, we mainly focus on the case where the threshold is the same for all $p$-values. The corresponding procedures, called \textit{thresholding based procedures}, are   of the form $R(\mbf{p})=\{1\leq i\leq m \telque p_i\leq t(\mbf{p})\}$, where the threshold $t(\cdot)\in[0,1]$ can depend on the data. 
  
 \begin{example}[Bonferroni procedure]
The Bonferroni procedure (of level $\alpha\in(0,1)$) rejects the hypotheses with a $p$-value smaller than $\alpha/m$. Hence, with our notation, it corresponds to the procedure $R(\mbf{p})=\{1\leq i\leq m\telque p_i\leq \alpha/m\}$. 
\end{example}

\subsection{Type I error rates}\label{sec:ER}

To evaluate the quality of a multiple testing procedure, various error rates have been proposed in the literature. According to the Neyman-Pearson approach, type I error rates are of primary interest.
These rates evaluate the importance of the null hypotheses wrongly rejected, that is, of the elements of the set $R(\mbf{p})\cap \cH_0(P)$. 
Nowadays, the most widely used type I error rates are the following. 
For a given procedure $R$,
\begin{itemize}
\item the \textit{$k$-family-wise error rate} ($\kFWER$) (see e.g. \citep{HT1987,RW2005,LR2005}) is defined as the probability that the procedure $R$ makes at least $k$ false rejections: for all $P \in \mathcal{P},$
\begin{align}
 \kFWER(R,P)= \P(|R(\mbf{p})\cap \cH_0(P)|\geq k)\label{def-kFWER},
 \end{align}
 where $k\in\{1,..., m\}$ is a pre-specified  parameter.
 In the particular case where $k=1$, this rate is simply called the \textit{family-wise error rate} and is denoted by FWER$(R,P)$.
\item the \textit{false discovery proportion} (FDP) (see e.g. \citep{See1968,BH1995,LR2005}) is defined as the proportion of errors in the set of the rejected hypotheses: for all $P \in \mathcal{P},$
\begin{align}
 \FDP(R(\mbf{p}),P)= \frac{|R(\mbf{p})\cap \cH_0(P)|}{|R(\mbf{p})|\vee 1}\label{def-FDP},
 \end{align}
where $|R(\mbf{p})|\vee 1$ denotes the maximum of $|R(\mbf{p})|$ and $1$. The role of the term ``$\vee 1$" in the denominator is to prevent from dividing by zero when $R$ makes no rejection.
Since the FDP is a random variable, it does not define an error rate. However, the following error rates can be derived from the FDP. First, the $\gamma$-upper-tail distribution of the FDP, defined as the probability that the FDP exceeds a given $\gamma$, that is, for all $P \in \mathcal{P},$
\begin{align}
\P( \FDP(R(\mbf{p}),P)>\gamma) \label{def-FDPtail},
 \end{align}
where $\gamma\in(0,1)$ is a pre-specified  parameter.
Second, the false discovery rate (FDR) \cite{BH1995}, defined as the expectation of the FDP: for all $ P \in \mathcal{P},$
\begin{align}
 \FDR(R,P)= \E [ \FDP(R(\mbf{p}),P)]=\E\bigg[\frac{|R(\mbf{p})\cap \cH_0(P)|}{|R(\mbf{p})|\vee 1}\bigg]\label{def-FDR}.
 \end{align}
\end{itemize}
Note that the probability in \eqref{def-FDPtail} is upper-bounded by a nominal level $\alpha\in(0,1)$ if and only if the $(1-\alpha)$-quantile of the FDP distribution is upper-bounded by $\gamma$. For instance, if the probability in \eqref{def-FDPtail} is upper-bounded by $\alpha=1/2$, this means that the median of the FDP is upper-bounded by $\gamma$. With some abuse, bounding the probability in \eqref{def-FDPtail}
is called ``controlling the FDP" from now on.

The choice of the 
type I error rate depends on the context. 
When controlling the $\kFWER$, we tolerate a fixed number $(k-1)$ of erroneous rejections. By contrast, a procedure controlling  \eqref{def-FDPtail}   tolerates a small proportion $\gamma$ of errors among the final rejections (from an intuitive point of view, it chooses  $k\simeq \gamma |R|$). This allows to increase the number of erroneous rejections as the number of rejections becomes large. 
Next, controlling the FDR has become popular because it is a simple error rate based on the FDP and because it came together with the simple Benjamini-Hochberg FDR controlling procedure \cite{BH1995} (some dependency structure assumptions are required, see Section~\ref{sec:FDR}). 
As a counterpart, controlling the FDR does not prevent the FDP from having large variations, so that any FDR control does not  necessarily have a clear interpretation in terms of the FDP (see the related discussion in Section~\ref{sec:discussFDP}).

\addtocounter{theorem}{-1}
 \begin{example}[Continued]
The Bonferroni procedure $R(\mbf{p})=\{1\leq i\leq m\telque p_i\leq \alpha/m\}$ satisfies the following:
$$
\E |R(\mbf{p})\cap \cH_0(P)|=\sum_{i\in\cH_0(P)} \P(p_i\leq \alpha/m)\leq \alpha m_0(P)/m\leq \alpha,
$$ 
which means that its expected number of false discoveries is below $\alpha$. Using Markov's inequality, this implies that $R(\mbf{p})$ makes no false discovery with probability at least $1-\alpha$, that is, for any $P\in\cP$, $\FWER(R,P)\leq \alpha$. This is the most classical example of type I error rate control.  
\end{example}

\begin{remark}[Case where $\cH_0(P)=\cH$]
\label{rem:weakcontrol}
For a distribution $P$ satisfying $\cH_0(P)=\cH$, that is when all null hypotheses are true, the FDP reduces to $\FDP(R(\mbf{p}),P)= \ind{|R(\mbf{p})|>0}$ and we have $\FWER(R,P)= \FDR(R,P)=\P( \FDP(R(\mbf{p}),P)>\gamma)=\P(|R(\mbf{p})|>0)$. Controlling the FWER (or equivalently the FDR) in this situation is sometimes called a ``weak" FWER control. 
\end{remark}

\begin{remark}[Case where all null hypotheses are equal:  $p$-value aggregation]
\label{rem:allnullthesame}
The general framework described in Section~\ref{sec:framework} includes the case where all null hypotheses are identical, that is, $\Theta_{0,i}=\Theta_0$ for all $i\in\{1,...,m\}$.  In this situation, all $p$-values test the same null $H_0:$ ``$P\in\Theta_0$" against some alternatives contained in $\Theta_0^c$. For instance, in the model selection framework of \cite{BHL2003,DR2006,VV2010}, each $p$-value is built with respect to a specific model contained in the alternative $\Theta_0^c$.
Since we have in that case $\cH_0(P)=\cH$ if $P\in \Theta_0$ and $\cH_0(P)=\emptyset$ otherwise, the three quantities $\FWER(R,P)$, $\FDR(R,P)$ and $\P( \FDP(R(\mbf{p}),P)>\gamma)$ are equal and take the value $\P(|R(\mbf{p})|>0)$ when $P\in \Theta_0$ and $0$ otherwise.
As a consequence, in the case where all null hypotheses are equal, controlling the FWER, the FDR or the FDP at level $\alpha$ is equivalent to the problem of combining $p$-values to build a \textit{single testing} for $H_0$ which is of level $\alpha$. 
In particular, from a procedure $R$ that controls the FWER at level $\alpha$ we can derive a single testing procedure of level $\alpha$ by rejecting $H_0$ whenever $R(\mbf{p})$ is not empty (that is, whenever $R(\mbf{p})$ rejects at least one hypothesis). This provides a way to aggregate $p$-values into one (single) test for $H_0$ which is ensured to be of level $\alpha$. 
As an illustration, the FWER controlling Bonferroni procedure $R=\{1\leq i \leq m \telque p_i\leq \alpha/m\}$ corresponds to the single test rejecting $H_0$ whenever $\min_{1\leq i\leq m}\{ p_{i}\}\leq \alpha/m$.  The Bonferroni combination of individual tests is well known and extensively used for adaptive testing (see, e.g., \cite{Spo1996,BHL2003,VV2010}). 
Some other examples of $p$-value aggregations will be presented further on, see Remark~\ref{rem:bul}. 
\end{remark}

\subsection{Goal}\label{sec:goal}

Let $\alpha\in(0,1)$ be a pre-specified nominal level (to be fixed once and for all throughout the paper).
The goal is to control the type I error rates defined above at level $\alpha$, for a large subset of distributions $\cP' \subset \cP$. 
That is, by taking one of the above error rate $\ER(R,P)$, we aim at finding a procedure $R$ such that 
\begin{align}
\forall P \in \mathcal{P}', \:\: \ER(R,P) \leq \alpha\label{def-FDRcontrol},
 \end{align}
for $\cP'\subset \cP$ as large as possible. Obviously, $R$ should depend on $\alpha$ but we omit this in the notation for short.
Similarly to the single testing case, taking $R=\emptyset$ will always ensure \eqref{def-FDRcontrol} with $\cP'=\cP$. This means that the type I error rate control is inseparable from the problem of maximizing the power. 
The probably most natural way to extend the notion of power from the single testing to the multiple testing setting is to consider the expected number of correct rejections, that is, $\E |\cH_1(P)\cap R|$. 
Throughout  the paper, we often encounter the case where two procedures $R$ and $R'$ satisfy $R'\subset R$ (almost surely) while they both ensure the control \eqref{def-FDRcontrol}. Then, the procedure $R$ is said \textit{less conservative} than $R'$. Obviously, this implies that $R$ is more powerful than $R'$.
This can be the case when, e.g., $R$ and $R'$ are thresholding-based procedures using respective thresholds $t$ and $t'$ satisfying $t\geq t'$ (almost surely). %
As a consequence, our goal is to find a procedure $R$ satisfying  \eqref{def-FDRcontrol} with a rejection set as large as possible.

Finally, let us emphasize that, in this paper, we aim at controlling \eqref{def-FDRcontrol} for any fixed $m\geq 2$ and not only when $m$ tends to infinity. That is, the setting is non-asymptotic in the parameter $m$.

\subsection{Overview of the paper}\label{sec:present}

The remainder of the paper is organized as follows: 
in Section~2, we present some general tools and concepts that are useful throughout the paper.
Section~\ref{sec:FDR}, \ref{sec:kFWER} and \ref{sec:FDP} present  FDR, $\kFWER$ and FDP controlling methodology, respectively,  where we try to give a large overview of classical methods in the literature. Besides, the paper is meant to have a scholarly form, accessible to a possibly non-specialist reader. In particular, all results are given together with a proof, which we aim to be as short and meaningful as possible. 

Furthermore, while this paper is mostly intended to be a review paper, some new contributions with respect to the existing multiple testing literature are given in 
Section~\ref{sec:kFWER} and \ref{sec:FDP}, by extending the results of  \cite{GS2010} for the $\kFWER$ control and the results of \cite{RW2007} for the FDP control, respectively.

\subsection{Quantile-binomial procedure }

In section~\ref{sec:FDP}, 
we introduce a novel procedure, called the \textit{quantile-binomial procedure} that  controls the FDP under independence of the $p$-values. This procedure can be defined as follows;

\begin{algorithm}[Quantile-binomial procedure]\label{algo-quant-binom}
Let for any $t\in[0,1]$ and for any $\l\in\{1,...,m\}$,
\begin{equation}\label{equ-quantilebinom}
q_\l(t) =\mbox{ the $(1-\alpha)$-quantile of $\mathcal{B}(m -\l+ \lfloor\gamma ({\l}-1)\rfloor + 1 , t)$},
\end{equation}
where $\mathcal{B}(\cdot,\cdot)$ denotes the binomial distribution and $ \lfloor\gamma ({\l}-1)\rfloor$ denotes the largest integer $n$ such that $n\leq \gamma ({\l}-1)$. 
Let $ p_{(1)}\leq ... \leq p_{(m)}$ be the order statistics of the $p$-values.
Then apply the following recursion:
\begin{itemize}
\item[\textbullet] Step $1$: if $q_1(p_{(1)})> \gamma$, stop and reject no hypothesis. Otherwise, go to step $2$;
\item[\textbullet] Step $\l\in\{2,...,m\}$: if $q_\l(p_{(\l)})> \gamma \l$, stop and reject the hypotheses corresponding to $p_{(1)}$, $\dots$, $p_{(\l-1)}$. Otherwise, go to step $\l+1$;
\item[\textbullet] Step $\l=m+1$, stop and reject all hypotheses. 
 \end{itemize}
 \end{algorithm}
 Equivalently, the above procedure can be defined as rejecting $H_{0,i}$ whenever $$\max_{p_{(\l)}\leq p_i} \{q_\l(p_{(\l)})/\l\}\leq \gamma.$$
 The rationale behind this algorithm is that at step $\l$, when rejecting the $\l$ null hypotheses corresponding to the $p$-values smaller than $p_{(\l)}$, the number of false discoveries behaves as if it was stochastically dominated by a binomial variable of parameter $(m -\l+ \lfloor\gamma ({\l}-1)\rfloor + 1, p_{(\l)})$.
Hence, by controlling the $(1-\alpha)$-quantile of the latter binomial variable at level $\gamma \l$, the $(1-\alpha)$-quantile of the FDP should be controlled by $\gamma$. The rigorous proof of the corresponding FDP control is given in Section~\ref{sec:FDP}, see Corollary~\ref{cor-LR2005-improved}. Finally, when controlling the median of the FDP, this procedure is related to the recent adaptive procedure of \cite{GBS2009}, as discussed in Section~\ref{sec:binom-BH}.

\section{Key concepts and tools} 

\subsection{Model assumptions}\label{sec:model}

Throughout this paper, we will consider several models. Each model corresponds to a specific assumption on the $p$-value family $\mbf{p}=\{p_i,1\leq i \leq m\}$ distribution. The first model, called the ``independent model" is defined as follows:
\begin{align}
\mathcal{P}^I=\:&\big\{ P \in \cP \telque 
(p_i(X))_{i\in\cH_0(P)} \mbox{ is a family of mutually independent}\nonumber\\
&\mbox{ variables and $(p_i(X))_{i\in\cH_0(P)}$ is independent of } (p_i(X))_{i\in\cH_1(P)}\big\}\label{modelindep}.
\end{align}
The second model uses a particular notion of positive dependence between the $p$-values, called ``weak positive regression dependency" (in short, ``weak PRDS"), which is a slightly weaker version of the PRDS assumption of \cite{BY2001}. 
To introduce the weak PRDS property, let us define a 
subset $D\subset[0,1]^m$ as \textit{nondecreasing} if for all $q,q' \in [0,1]^m$
such that $\forall i\in\{1,...,m\}$, $q_i\leq q_i'$, we have $q' \in D$ when $q\in D$. 
\begin{definition}[Weak PRDS $p$-value family]\label{def:weakPRDS}
The family $\mbf{p}$ is said to be  \textit{weak PRDS on $\cH_0(P)$} if 
for any $i_0\in\cH_0(P)$ and for any measurable nondecreasing set $D\subset[0,1]^m$\,, the function 
$ u \mapsto \Proba(\mbf{p} \in D\cond p_{i_0}\leq u)$
is nondecreasing on the set $\{u\in[0,1]\telque \P(p_{i_0}\leq u)>0\}.$
\end{definition}
The only difference between the weak PRDS assumption and the ``regular" PRDS assumption defined in \cite{BY2001} is that the latter assumes ``$ u \mapsto \Proba(\mbf{p} \in D\cond p_{i_0}=u)$  nondecreasing", instead of ``$ u \mapsto \Proba(\mbf{p} \in D\cond p_{i_0}\leq u)$ nondecreasing". Weak PRDS is a weaker assumption, as shown for instance in the proof of Proposition~3.6 in \cite{BR2008}. We can now define the second model, where the $p$-values have weak PRDS dependency: 
\begin{align}
\mathcal{P}^{pos}=\:&\big\{ P \in \cP \telque 
\mbf{p}(X) \mbox{ is weak PRDS on $\cH_0(P)$} \big\}\label{modeposdep}.
\end{align}
It is not difficult to see that $\cP^I\subset \cP^{pos}$ because when $P\in \cP^I$,  $p_{i_0}$ is independent of $(p_i)_{i\neq i_0}$
 for any $i_0\in\cH_0(P)$. Furthermore, we refer to the general case of $P\in\cP$ (without any additional restriction) as the  ``arbitrary dependence case". 

As an illustration, in the one-sided Gaussian testing framework of Example~\ref{sec:vectorpos}, the PRDS assumption (regular and thus also weak) is satisfied as soon as the covariance matrix $\Sigma(P)$ has nonnegative entries, as shown in \cite{BY2001} (note that this is not true anymore for two-sided tests, as proved in the latter reference).

\subsection{Dirac configurations}\label{sec:DU}

If we want to check whether a procedure satisfies a type I error rate control \eqref{def-FDRcontrol}, particularly simple $p$-value distributions (or  ``configurations") are as follows:
\begin{itemize}
\item[-]  ``Dirac configurations":  the $p$-values of $\cH_1(P)$ are equal to zero (without any assumption on the $p$-values of $\cH_0(P)$);
\item[-]  ``Dirac-uniform configuration" (see \cite{FDR2007}):  the Dirac configuration for which the variables 
$(p_i)_{i \in \cH_0(P)}$ are i.i.d. uniform. 
\end{itemize}
These configurations can be seen as the asymptotic $p$-value family distribution where the sample size available to perform each test tends to infinity, while the number $m$ of tests is kept fixed (see the examples of Section~\ref{sec:exTM}).
This situation does not fall into the classical multiple testing framework where the number of tests is much larger than the sample size. Besides, there is no multiple testing problem in these configurations because the true nulls are  perfectly separated from the false null (almost surely). 
However,  these special configurations are still interesting, because they sometimes have the property to be the distributions for which the type I error rate is the largest. In that case, they are called the ``least favorable configurations" (see \cite{FDR2007}). This generally requires that the multiple testing procedure and the error rate under consideration have special monotonic properties (see \cite{FDR2009,RV2010}). In this case, proving the type I error rate control for the Dirac configurations is sufficient to state \eqref{def-FDRcontrol} and thus appears to be very useful.

\subsection{Algorithms}\label{sec:genmethod}

To derive \eqref{def-FDRcontrol}, a generic method that emerged from the multiple testing literature is as follows:
\begin{enumerate}
\item  start with a family $(R_\kappa)_\kappa$ of procedures depending on an external parameter $\kappa$;
\item  find a set of values of $\kappa$ for which $R_\kappa$ satisfies \eqref{def-FDRcontrol};
\item  take among these values the $\kappa$ that makes $R_\kappa$ the ``largest".
\end{enumerate}
The latter is designed to maintain the control of the type I error rate while maximizing the rejection set.
As we will see in Section~\ref{sec:FDR} ($\kappa$ is a threshold $t$), Section~\ref{sec:kFWER} ($\kappa$ is a subset $\cC$ of $\cH$) and Section~\ref{sec:FDP} ($\kappa$ is a rejection number $\l$), this gives rise to the so-called ``step-up" and ``step-down" algorithms, which are very classical instances of type I error rate controlling procedures.  

\subsection{Adaptive control}

A way to increase the power of type I error rate controlling procedures is to learn (from the data) part of the unknown distribution $P$ in order to make more rejections. 
This approach is called ``adaptive type I error rate control".
Since the resulting procedure uses the data twice, the main challenge is often to show that it maintains the type I error control \eqref{def-FDRcontrol}. In this paper, we will discuss adaptivity with respect to the parameter $m_0(P)=|\cH_0(P)|$ for the $\FDR$ in Section~\ref{sec:ALSU}. The procedures presented in Section~\ref{sec:kFWER} (resp. Section~\ref{sec:FDP}) for controlling the $\kFWER$ (resp. FDP) will be also adaptive 
to $m_0(P)$,  but in a maybe more implicit way. Some of them will be additionally adaptive with respect to the dependency structure between the $p$-values. Let us finally note that some other work studied the adaptivity to the alternative distributions of the $p$-values (see \cite{WR2006,RDV2006,RW2009}).

\section{FDR control}\label{sec:FDR}

After the seminal work of Benjamini and Hochberg \cite{BH1995}, many studies have investigated the FDR controlling issue. We provide in this section a survey of some of these approaches. 

\subsection{Thresholding based procedures}

Let us start from thresholding type multiple-testing procedures $$R_t=\{1\leq i\leq m\telque p_i\leq t(\mbf{p})\},$$ with a threshold ${t}(\cdot)\in [0,1]$ possibly depending on the $p$-values. We want to find $t$ such that the corresponding multiple testing procedure $R_t$ controls the FDR at level $\alpha$ under the model $\cP^{pos}$, by following the general method explained in Section~\ref{sec:genmethod}.
We start with the following simple decomposition of the false discovery rate of $R_t$: 
\begin{equation}
\label{FDR-form1}
\FDR(R_t,P)=\alpha m^{-1} \sum_{i\in\cH_0(P)} \E\bigg[\frac{\ind{p_i\leq t(\mbf{p})}}{\alpha\: \G(\mbf{p},t(\mbf{p})) \vee (\alpha/m)}\bigg],
\end{equation}
where  $\G(\mathbf{p},u)=m^{-1} \sum_{i=1}^m \ind{p_i\leq  u}$ denotes the empirical c.d.f. of the $p$-value family $\mbf{p}=\{p_i,1\leq i \leq m\}$ taken at a threshold $u\in[0,1]$. 

In order to upper-bound the expectation in the RHS of \eqref{FDR-form1}, let us consider the following informal reasoning: if $t$ and $\G$ were deterministic, this expectation would be smaller than $t/(\alpha\: \G(\mbf{p},t))$ and thus smaller than $1$ by taking a threshold $t$ such that $t\leq \alpha\: \G(\mbf{p},t)$. This motivates the introduction of the following set of thresholds: 
\begin{equation}
\label{SCset}
\mathcal{T}(\mbf{p}) =\{ u \in [0,1] \telque \wh{\mathbb{G}}( \mbf{p},u) \geq u/ \alpha\}. 
\end{equation}
With different notation, the latter was introduced in \cite{BR2008,FDR2009}. Here, any threshold $t\in \mathcal{T}(\mbf{p})$ is said ``self-consistent" because it corresponds to a procedure $R_t=\{1\leq i\leq m\telque p_i\leq t\}$ which is ``self-consistent" according to the definition given in \cite{BR2008}, that is, $R_t\subset \{1\leq i\leq m\telque p_i\leq \alpha |R_t|/m\}$.
It is important to note that the set $\mathcal{T}(\mbf{p})$ only depends on the $p$-value family (and on $\alpha$) so that self-consistent thresholds can be easily chosen in practice. As an illustration, we depict the set $\mathcal{T}(\mbf{p})$ in Figure~\ref{fig_SC}  for a particular realization of the $p$-value family.

 \begin{figure}[htbp]
\begin{center}
\includegraphics[scale=0.3,angle=-90]{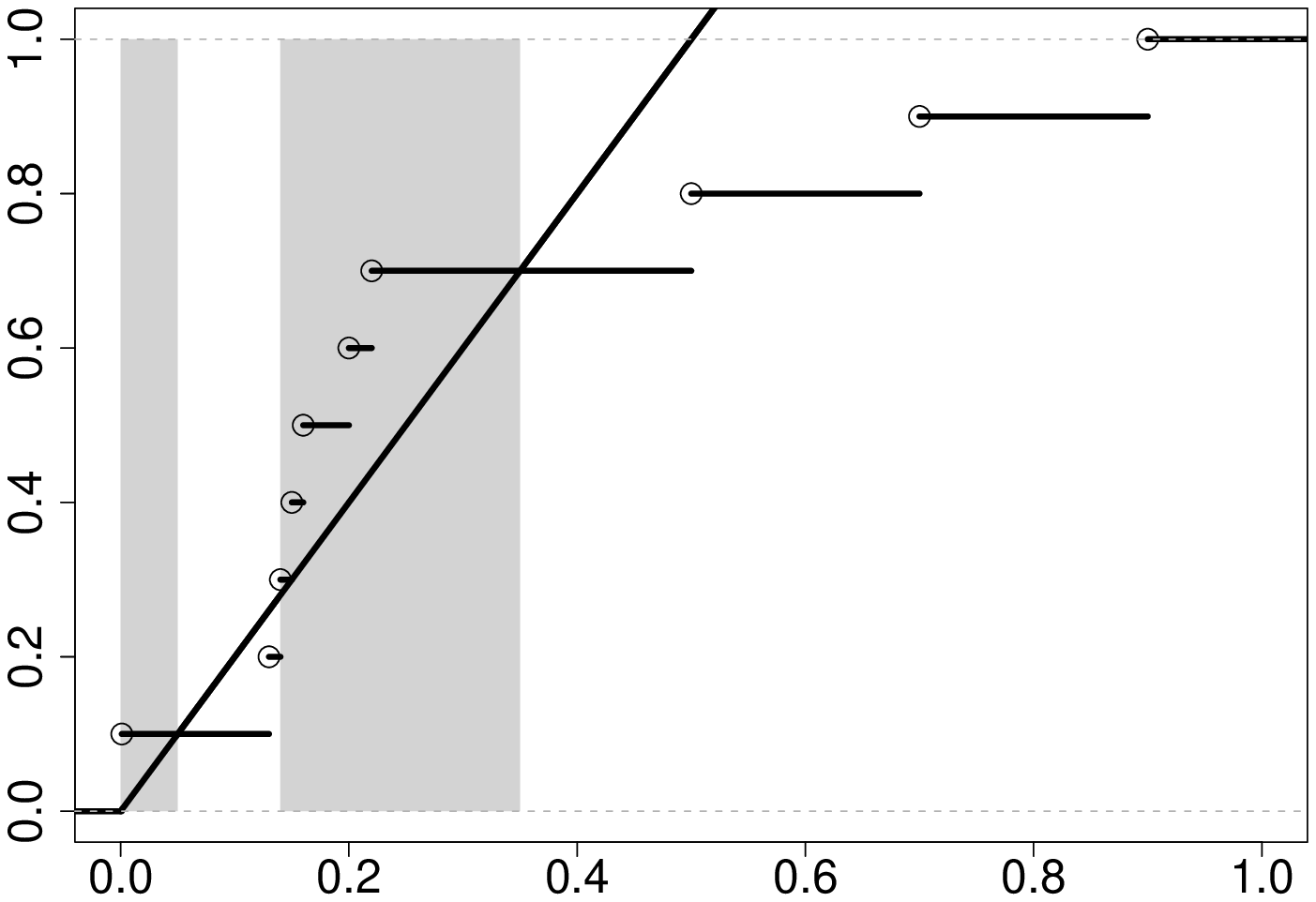}
\caption{The $p$-value e.c.d.f $\wh{\mathbb{G}}( \mbf{p},u)$ and $u/\alpha$ are plotted as functions of $u\in[0,1]$. The points $u$ belonging to the set $\mathcal{T}(\mbf{p})$ lie on the X-axis of the gray area. $m=10$; $\alpha=0.5$. 
}
\label{fig_SC}
\end{center}
\end{figure}

Now, let us choose a self-consistent threshold $t(\mbf{p})\in \mathcal{T}(\mbf{p})$. By using the decomposition \eqref{FDR-form1}, we obtain the following upper-bound: 
\begin{equation}
\label{FDR-form2}
\FDR(R_t,P)\leq \alpha m^{-1} \sum_{i\in\cH_0(P)} \E\bigg[\frac{\ind{p_i\leq t(\mbf{p})}}{t(\mbf{p}) \vee (\alpha/m) }\bigg]\leq \alpha m^{-1} \sum_{i\in\cH_0(P)} \E\bigg[\frac{\ind{p_i\leq t(\mbf{p})}}{t(\mbf{p}) }\bigg],
\end{equation}
with the convention $\frac{0}{0}=0$. Since by \eqref{equ_defpvalue},  we have $p_i(x)>0$ for $P$-almost every $x$ when $i\in\cH_0(P)$, the denominator inside the expectation of the RHS of \eqref{FDR-form2}  can only be zero when the numerator is also zero and therefore when the ratio is zero. Next, the following purely probabilistic lemma holds (see a proof in Appendix~A of \cite{BR2008} for instance):  
\begin{lemma}\label{lemmaDC}
Let $U$ be a nonnegative random variable which is stochastically lower bounded by a uniform distribution, i.e., $\P(U\leq u)\leq u$ for any $u\in[0,1]$. Then the following inequality holds:
\begin{equation}
\label{conditionDC}
 \E\bigg[\frac{\ind{U\leq V}}{V}\bigg]\leq 1\,,
\end{equation}
for any nonnegative random variable $V$ satisfying either of the two following conditions:
\begin{itemize}
\item[(i)] $V=g(U)$ where $g:\mathbb{R}^+ \rightarrow \mathbb{R}^+$ is non-increasing,
\item[(ii)] the conditional distribution of $V$ conditionally on $U\leq u$ is stochastically decreasing in $u$, that is, 
$\forall v\geq 0$, $u\mapsto \P(V<v \cond U\leq u)$ is nondecreasing on $\{u\in[0,1]\telque \P(U\leq u)>0\}$. 
\end{itemize}
\end{lemma}
A consequence of the previous lemma in combination with \eqref{FDR-form2} is that the FDR is controlled at level $\alpha m_0(P)/m$ as soon as $V=t(\mbf{p})$ satisfies (ii) with $U=p_i$. For the latter to be true, we should make the distributional assumption $ P \in \mathcal{P}^{pos}$ and  add the assumption that the threshold ${t}(\cdot)$ is non-increasing with respect to each $p$-value, that is, for all $q,q' \in [0,1]^m$, we have $t(q)\leq t(q')$ as soon as for all $1\leq i\leq m$, $q'_i\leq q_i$.
By using the latter, we easily check that the set $$D=\{q\in[0,1]^m\telque {t}(q)<v\}$$ is a nondecreasing measurable set of $[0,1]^m$, for any $v\geq 0$. Thus, the weak PRDS condition defined in Section~\ref{sec:model} provides (ii) with $U=p_i$ and $V=t(\mbf{p})$ and thus also \eqref{conditionDC}. 
Summing up, we obtained the following result, which appeared in \cite{BR2008}:
\begin{theorem}
\label{thFDR1}
Consider a thresholding type multiple testing procedure $R_t$ based on a threshold $t(\cdot)$ satisfying the two following conditions:
\begin{itemize}
\item[-]  $t(\cdot)$ is self-consistent, i.e., such that for all $q \in [0,1]^m$, $t(q)\in \mathcal{T}(q)$ (where $\mathcal{T}(\cdot)$ is defined by \eqref{SCset}) 
\item[-]  $t(\cdot)$ is coordinate-wise non-increasing, i.e., satisfying that for all $q,q' \in [0,1]^m$ with $q'_i\leq q_i$ for all $1\leq i\leq m$, we have $t(q)\leq t(q')$. 
\end{itemize}
Then, for any $P \in \mathcal{P}^{pos}$, $\FDR(R_t,P)\leq \alpha m_0(P)/m\leq \alpha$.
\end{theorem}

\begin{remark}
If we want to state the FDR control of Theorem~\ref{thFDR1} only for $P \in \mathcal{P}^{I}$ without using the PRDS property, we can use Lemma~\ref{lemmaDC} (i)  \textit{conditionally on $\mbf{p}_{-i}=(p_j,j\neq i)\in[0,1]^{m-1}$}, by taking $V=t(U,\mbf{p}_{-i})$ and  $U=p_i$, because $p_i$ is independent of $\mbf{p}_{-i}$ when $P \in \mathcal{P}^{I}$.
\end{remark}

\subsection{Linear step-up procedures}\label{sec:LSU}

From Theorem~\ref{thFDR1}, under the weak PRDS assumption on the $p$-value dependence structure, any algorithm giving as output a self-consistent and non-increasing threshold $t(\cdot)$ leads to a correct FDR control. 
As explained in Section~\ref{sec:goal} and Section~\ref{sec:genmethod}, for the same FDR control we want to get a procedure with a rejection set as large as possible. Hence, it is natural to choose the following threshold:
\begin{align}
\label{equ_SU}
t^{su}(\mbf{p})&=\max\{ \mathcal{T}(\mbf{p}) \} \\
&=\max \{ u \in \{\alpha k /m, 0\leq k \leq m\} \telque \wh{\mathbb{G}}( \mbf{p},u) \geq u/ \alpha\}\nonumber\\
&=\alpha /m \:\times \max \{  0\leq k \leq m \telque p_{(k)} \leq  \alpha k/m\},\label{equ_SU-BHform}
\end{align}
where $p_{(1)}\leq ... \leq p_{(m)}$ ($p_{(0)}=0$) denote the order statistics of the $p$-value family.
This choice was made in \cite{BH1995} and is usually called \textit{linear step-up} or ``Benjamini-Hochberg" thresholding. 
One should notice that the maximum in \eqref{equ_SU} exists because the set $\mathcal{T}(\mbf{p})$ contains $0$, is upper-bounded by $1$ and because the e.c.d.f. is a non-decreasing function (the right-continuity is not needed). It is also easy to check that the maximum $u=t^{su}(\mbf{p})$ satisfies the equality $\wh{\mathbb{G}}( \mbf{p},u) = u/ \alpha$, so that $t^{su}(\mbf{p})$ can be seen as the largest crossing point between between $u\mapsto \wh{\mathbb{G}}( \mbf{p},u)$ and $u\mapsto u/ \alpha$, 
see the left-side of Figure~\ref{fig_LSU}. The latter equality also implies that $t^{su}(\mbf{p}) \in \{\alpha k /m, 0\leq k \leq m\}$, which, combined with the so-called switching relation $$m\:\wh{\mathbb{G}}( \mbf{p},\alpha k/m) \geq k \Longleftrightarrow p_{(k)}\leq \alpha k/m,$$ gives rise to the second formulation \eqref{equ_SU-BHform}. The latter is illustrated in the right-side of Figure~\ref{fig_LSU}. 
The formulation \eqref{equ_SU-BHform} corresponds to the original expression of \cite{BH1995} while \eqref{equ_SU} is to be found for instance in \cite{GW2002}. 
Moreover, it is worth noticing that the procedure $R_{t^{su}}$ using the thresholding $t^{su}(\mbf{p})$ is also equal to  $\{1\leq i \leq m\telque p_i\leq t^{su}(\mbf{p})\vee \alpha /m\}$, so that it can be interpreted as an intermediate thresholding between the non-corrected procedure using $t=\alpha$  and the Bonferroni procedure using $t=\alpha/m$.

 \begin{figure}[htbp]
\begin{center}
\includegraphics[scale=0.3,angle=-90]{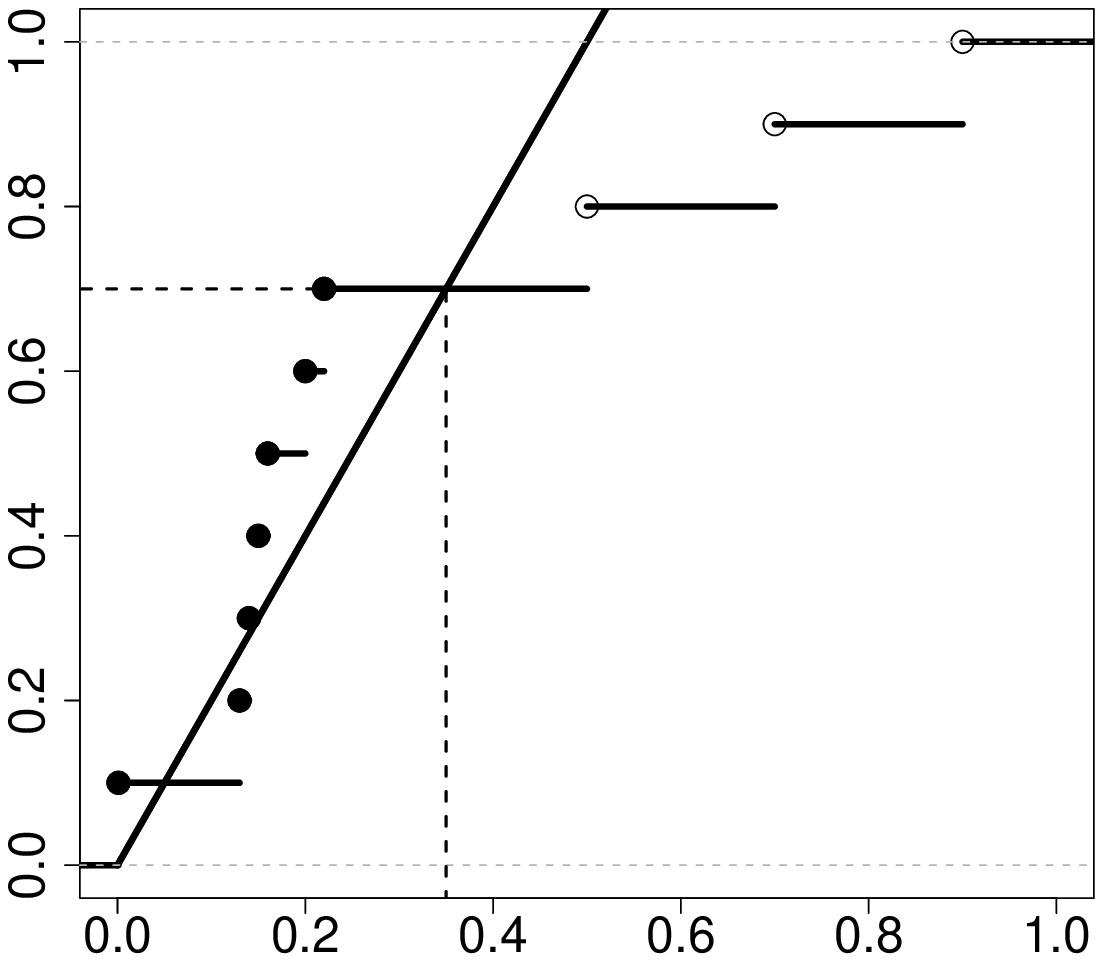}
\includegraphics[scale=0.3,angle=-90]{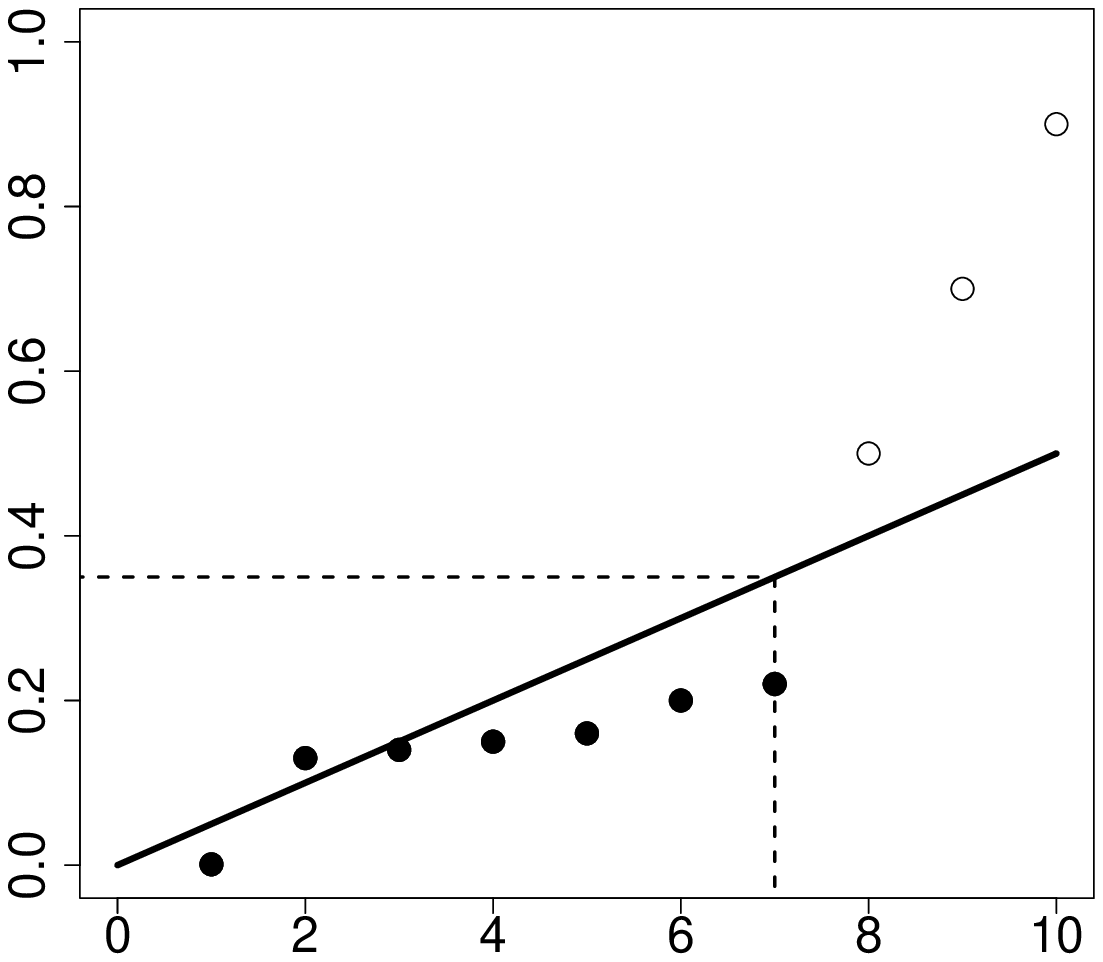}
\begin{pspicture}(0,0)(0,0)
\rput[B](-9.8,-5.7){$t^{su}$}
\rput[B](-6.2,-3.6){$t^{su}$}
\rput[B](-12.6,-2.1){{$\wh{\mathbb{G}}( t^{su} )$}}
\rput[B](-2,-5.7){\tiny{$|R_{t^{su}}|$}}
\end{pspicture}
\caption{The two dual pictorial representations of the Benjamini-Hochberg linear step-up procedure. 
Left: c.d.f. of the $p$-values, the solid line has for slope $\alpha^{-1}$.
Right: ordered $p$-values, the solid line has for slope $\alpha/m$.
In both pictures, the filled points represent $p$-values that corresponds to the rejected hypotheses.
 $m=10$; $\alpha=0.5$.}
\label{fig_LSU}
\end{center}
\end{figure}

Clearly, $t^{su}(\cdot)$ is coordinate-wise non-increasing and self-consistent. Therefore, Theorem~\ref{thFDR1} shows that for any $P \in \mathcal{P}^{pos}$, $\FDR(R_{t^{su}},P)\leq \alpha m_0(P)/m$. As a matter of fact,  as soon as \eqref{equ_defpvalue} holds with an equality, we can prove that for any $P \in \mathcal{P}^{I}$, the equality $\FDR(R_{t^{su}},P)=\alpha m_0(P)/m$ holds, by using a surprisingly direct argument. 
Let $\mbf{p}_{0,-i}$ denote the $p$-value family where $p_i$ has been replaced by $0$, and observe that the following statements are equivalent, for any realization of the $p$-values:
\begin{itemize}
\item[(i)]  $p_i\leq t^{su}(\mbf{p}_{0,-i})$
\item[(ii)] $ \wh{\mathbb{G}}\big( \mbf{p}_{0,-i}, t^{su}(\mbf{p}_{0,-i}) \big) \leq \wh{\mathbb{G}}\big( \mbf{p}, t^{su}(\mbf{p}_{0,-i}) \big)$
\item[(iii)]   $ t^{su}(\mbf{p}_{0,-i})/\alpha \leq \wh{\mathbb{G}}\big( \mbf{p}, t^{su}(\mbf{p}_{0,-i}) \big)  $
\item[(iv)]  $ t^{su}(\mbf{p}_{0,-i})\leq t^{su}(\mbf{p}).$
\end{itemize}
The equivalence between (i) and (ii) is straightforward from the defintion of $\wh{\mathbb{G}}(\cdot,\cdot)$. The equivalence between (ii) and (iii)  follows from $\wh{\mathbb{G}}\big( \mbf{p}_{0,-i}, t^{su}(\mbf{p}_{0,-i}) \big)= t^{su}(\mbf{p}_{0,-i})/\alpha$, because $t=t^{su}(\mbf{p}_{0,-i})$ is a crossing point between $\wh{\mathbb{G}}( \mbf{p}_{0,-i}, t )$ and $t/\alpha$. The equivalence between (iii) and (iv)  comes from the definition of $t^{su}(\mbf{p})$ together with  $t^{su}(\mbf{p}_{0,-i})\leq t^{su}(\mbf{p})\Longleftrightarrow t^{su}(\mbf{p}_{0,-i}) = t^{su}(\mbf{p})$, the latter coming from the  non-increasing property of  $t^{su}(\cdot)$. As a consequence,
 \begin{equation}\label{equ-equiv}\{p_i\leq t^{su}(\mbf{p}_{0,-i})\}=\{ p_i \leq t^{su}(\mbf{p})\},\end{equation}
 with $t^{su}(\mbf{p}_{0,-i}) = t^{su}(\mbf{p})$ on these events. 
Therefore, using \eqref{equ-equiv} and the first decomposition \eqref{FDR-form1} of the FDR, we derive the following equalities:
\begin{align*}
\FDR(R_{t^{su}},P)&=\alpha m^{-1} \sum_{i\in\cH_0(P)} \E\bigg[\frac{\ind{p_i\leq t^{su}(\mbf{p})}}{\alpha\: \G(\mbf{p},t^{su}(\mbf{p})) \vee (\alpha/m)}\bigg]\\
&= \alpha m^{-1} \sum_{i\in\cH_0(P)} \E\bigg[\frac{\ind{p_i\leq t^{su}(\mbf{p})}}{t^{su}(\mbf{p})}\bigg]\\
&= \alpha m^{-1} \sum_{i\in\cH_0(P)} \E\bigg[\frac{\ind{p_i\leq t^{su}(\mbf{p}_{0,-i})}}{t^{su}(\mbf{p}_{0,-i})}\bigg]\\
&= \alpha m^{-1} \sum_{i\in\cH_0(P)} \E\bigg[ t^{su}(\mbf{p}_{0,-i})^{-1} \E\big( \ind{p_i\leq t^{su}(\mbf{p}_{0,-i})}\big| \mbf{p}_{0,-i}\big)\bigg]\\
&=\alpha m_0(P)/m,
\end{align*}
where we assumed in the last equality both that 
$P \in \cP^I$ and 
condition \eqref{equ_defpvalue} holds with equality.
To sum up, we have proved in this section the following result.
\begin{theorem}
\label{thFDR2}
Consider the linear step-up procedure $R_{t^{su}}$ using the threshold defined in \eqref{equ_SU}. Then, for any $P \in \mathcal{P}^{pos}$, $\FDR(R_{t^{su}},P)\leq \alpha m_0(P)/m$. Moreover, the latter is an equality if $P\in\mathcal{P}^{I}$ and  \eqref{equ_defpvalue} holds with equality.
\end{theorem}
This theorem is due to \cite{BH1995,BY2001}. The short proof mentioned above has been independently given in \cite{FZ2006,RW2009,FDR2009}. Theorem~\ref{thFDR2} proves that the inequality ``$\forall P \in \mathcal{P}^{pos}$, $\FDR(R_{t^{su}},P)\leq \alpha$" is sharp as soon as \eqref{equ_defpvalue} holds with equality 
and 
there exists $P\in\cP^I$ such that $\cH_0(P)=\cH$, that is, $\cap_{i\in\cH} \Theta_{0,i} \cap \cP^I\neq \emptyset$.

Other instances of self-consistent procedures include linear ``step-up-down" procedures as defined in \cite{Sar2002}. Theorem~\ref{thFDR1} establishes that the FDR control also holds for these procedures, as proved in \cite{BR2008,FDR2009}. 

\subsection{Adaptive linear step-up procedures}\label{sec:ALSU}

In this section we denote by $\pi_0(P)$ the proportion $m_0(P)/m$ of hypotheses that are true for $P$. Since we aim at controlling the FDR at level $\alpha$ and not at level $\alpha \pi_0(P)$, Theorem~\ref{thFDR2} shows that there is a potential power loss when using $t^{su}$ when the proportion $\pi_0(P)$ is small. A first idea is to use the linear step-up procedure at level $\alpha^\star=\min(\alpha/\pi_0(P),1)$, that is, corresponding to the threshold
\begin{align}
\label{equ_OASU-sature}
t^*(\mbf{p})&=  \max \big\{ u \in [0,1] \telque \wh{\mathbb{G}}( \mbf{p},u) \geq u /\alpha^\star \big\} \\
\label{equ_OASU}
&=\max \big\{ u \in [0,1] \telque \wh{\mathbb{G}}( \mbf{p},u) \geq u \:\pi_0(P)/\alpha \big\}.
\end{align}
Note that \eqref{equ_OASU-sature} and \eqref{equ_OASU} are equal because when $\alpha\geq \pi_0(P)$, the maximum is $1$ in the two formulas. 
From  Theorem~\ref{thFDR2}, threshold \eqref{equ_OASU} provides a FDR smaller than $\alpha^\star \pi_0(P)\leq\alpha$ for $P\in \mathcal{P}^{pos}$ and a FDR equal to $\alpha$ when $P\in\mathcal{P}^{I}$, \eqref{equ_defpvalue} holds with equality and $\alpha\leq \pi_0(P)$. Unfortunately, since $P$ is unknown, so is $\pi_0(P)$ and thus the threshold \eqref{equ_OASU} is an unobservable ``oracle" threshold. 

An interesting challenge is to estimate $\pi_0(P)$ within \eqref{equ_OASU} while still rigorously controlling the FDR at level $\alpha$, despite the additional fluctuations added by the $\pi_0(P)$-estimation. This problem, called $\pi_0(P)$-adaptive FDR control, has received a growing attention in the last decade, see e.g. \cite{BH2000,STS2004,Black2004,GW2004,BKY2006,Neu2008,Sar2008,BR2009}. 
To investigate this issue, a natural idea is to consider a modified linear step-procedure using the threshold
\begin{equation}
\label{equ_ASU}
t^{su}_f(\mbf{p})= \max \big\{ u \in [0,1] \telque \wh{\mathbb{G}}( \mbf{p},u) \geq u /\big(\alpha\: f(\mbf{p})\big) \big\}.
\end{equation}
where $f(\mbf{p})>0$ is an estimator of $(\pi_0(P))^{-1}$ to be chosen. The latter is called \textit{adaptive linear step-up procedure}. 
It is sometimes additionally said ``plug in", because \eqref{equ_ASU} corresponds to \eqref{equ_OASU} in which we have ``plugged" an estimator of $(\pi_0(P))^{-1}$. Other types of adaptive procedures can be defined, see Remark~\ref{rem:onestage} below.

 We describe now a way to choose $f$ so that the control $\FDR(R_{t^{su}_f},P)\leq \alpha$ still holds. However,   we only focus on the case where the $p$-values are independent, that is, $P\in\mathcal{P}^I$. This restriction is usual in studies providing an adaptive FDR control.
First, to keep the non-increasing property of the threshold $t^{su}_f(\cdot)$, we assume that $f(\cdot)$ is coordinate-wise non-increasing. Second, using  techniques similar to those of Section~\ref{sec:LSU}, 
we can write for any $P\in\mathcal{P}^I$,
 \begin{align} 
 \FDR(R_{t^{su}_f},P)
&\leq \alpha m^{-1} \sum_{i\in\cH_0(P)} \E\bigg[\frac{\ind{p_i\leq t^{su}_f(\mbf{p})}}{t^{su}_f(\mbf{p})}f(\mbf{p})\bigg]\nonumber\\
&\leq \alpha m^{-1} \sum_{i\in\cH_0(P)} \E\bigg[\frac{\ind{p_i\leq t^{su}_f(\mbf{p})}}{t^{su}_f(\mbf{p})}f(\mbf{p}_{0,-i})\bigg]\nonumber\\
&= \alpha m^{-1} \sum_{i\in\cH_0(P)} \E\bigg[ f(\mbf{p}_{0,-i}) \E\bigg[ \frac{\ind{p_i\leq t^{su}_f(\mbf{p})}}{t^{su}_f(\mbf{p})}\bigg| \mbf{p}_{0,-i}\bigg]\bigg]\nonumber\\
&\leq \alpha m^{-1} \sum_{i\in\cH_0(P)} \E\big[f(\mbf{p}_{0,-i})\big],\label{equ-ALSU}
\end{align}
 where we used Lemma~\ref{conditionDC} (i) in the last inequality (conditionally on the $p$-values of $(p_{j},j\neq i)$, because $f$ is coordinate-wise non-increasing). 
Additionally assuming that $f(\cdot)$ is permutation invariant, we can upper-bound the RHS of \eqref{equ-ALSU} by using the Dirac-uniform configuration 
because $f(\cdot)$ is non-increasing. 
This gives rise to the following result.
\begin{theorem}
\label{thAFDR}
Consider the adaptive linear step-up procedure $R_{t^{su}_f}$ with a threshold defined in \eqref{equ_ASU} using a $(\pi_0(P))^{-1}$-estimator $f$ satisfying the following properties:
\begin{itemize}
\item $f(\cdot)$ is coordinate-wise non-increasing, that is, for all $q,q' \in [0,1]^m$ with for all $1\leq i\leq m$, $q'_i\leq q_i$, we have $f(q)\leq f(q')$;
\item $f(\cdot)$ is permutation invariant, that is, for any permutation $\sigma$ of $\{1,...,m\}$, $\forall q\in[0,1]^m$, $f(q_1,...,q_m)=f(q_{\sigma(1)},...,q_{\sigma(m)})$;
\item $f$ satisfies
\begin{equation}\label{equ_estimcond}\forall m_0\in\{1,...,m\}, \:\:\:\E_{\mbf{p}\sim DU(m_0-1,m)}( f(\mbf{p})) \leq m/m_0,\end{equation} 
where $DU(k,m)$ denotes the Dirac-uniform distribution on $[0,1]^m$ for which the $k$ first coordinates are i.i.d. uniform on $(0,1)$ and the remaining coordinates are equal to $0$.  
\end{itemize}
Then, for any $P \in \mathcal{P}^{I}$, $\FDR(R_{t^{su}_f},P)\leq \alpha$.
\end{theorem}

The method leading to the upper-bound \eqref{equ-ALSU} was investigated in \cite{BKY2006} and described latter in detail in \cite{BR2009}. The simpler result presented in Theorem~\ref{thAFDR} appeared in \cite{BR2009}. It uses the Dirac-uniform configuration as a least favorable configuration for the FDR. This kind of reasoning has been also used in \cite{FDR2009}. 

Let us now consider the problem of finding a ``correct" estimator $f$ of $(\pi_0(P))^{-1}$. This issue has an interest in its own right and many studies investigated it since the first attempt in \cite{SS1982} (see for instance the references in \cite{CR2010}). Here, we only deal with this problem from the FDR control point of view, by providing two families of estimators that satisfy the assumptions of Theorem~\ref{thAFDR}. First, define the ``Storey-type" estimators, which are of the form
$$
f_1(\mbf{p}) = \frac{m(1-\lambda)}{\sum_{i=1}^m \ind{p_i>\lambda} +1},
$$
for $\lambda\in (0,1)$ ($\lambda$ not depending on $\mbf{p}$).
It is clearly non-increasing and permutation invariant. Moreover, we can check that $f_1$ satisfies \eqref{equ_estimcond}: for any $m_0\in\{1,...,m\}$, considering $(U_i)_{1\leq i\leq m_0-1}$ i.i.d. uniform on $(0,1)$,
$$
\E_{\mbf{p}\sim DU(m_0-1,m)}( f_1(\mbf{p})) = \frac{m}{m_0}\E \bigg[\frac{m_0(1-\lambda)}{\sum_{i=1}^{m_0-1} \ind{U_i>\lambda} +1} \bigg]\leq \frac{m}{m_0},
$$
because for any $k\geq 2$, $q\in(0,1)$ and for $Y$ having a binomial distribution with parameters $(k-1,q)$, 
 we have $\E((1+Y)^{-1})\leq (qk)^{-1}$, as stated e.g. in \cite{BKY2006}. This type of estimator has been introduced in \cite{Storey2002} and proved to lead to a correct FDR control in \cite{STS2004,BKY2006}.

The second family of estimators satisfying the assumptions of Theorem~\ref{thAFDR} is the ``quantile-type" family, defined by
$$
f_2(\mbf{p}) =  \frac{m(1-p_{(k_0)})}{m-k_0 +1},
$$
for $k_0\in\{1,...,m\}$ ($k_0$ not depending on $\mbf{p}$). 
The latter may be seen as Storey-type estimators using a data-dependent $\lambda=p_{(k_0)}$.
Clearly, $f_2(\cdot)$ is non-increasing and permutation-invariant. 
Additionally, $f_2(\cdot)$ enjoys \eqref{equ_estimcond} because for any $m_0\in\{1,...,m\}$, considering $(U_i)_{1\leq i\leq m_0-1}$ i.i.d. uniform on $(0,1)$ ordered as $U_{(1)}\leq ... \leq U_{(m_0-1)}$,
\begin{align*}
\E_{\mbf{p}\sim DU(m_0-1,m)}( f_2(\mbf{p})) &=  \E\bigg[ \frac{m(1-U_{(k_0-m+m_0-1)})}{m-k_0 +1}\bigg] =  \frac{m(1-\E[U_{(k_0-m+m_0-1)}])}{m-k_0 +1}\\
&=    \frac{m(1-(k_0-m+m_0-1)_+/m_0)}{m-k_0 +1} \leq \frac{m}{m_0},
\end{align*}
by using the convention $U_{(j)}=0$ when $j\leq 0$. These quantile type estimators have been proved to lead to a correct FDR control in \cite{BKY2006}. The  simple proof above was given in \cite{BR2009}.

Which choice should we make for $\lambda$ or $k_0$? Using extensive simulations (including other type of adaptive procedures), it was recommended in \cite{BR2009} to choose as estimator $f_1$ with $\lambda$ close to $\alpha$, because the corresponding procedure shows a ``good" power under independence while it maintains a correct FDR control under positive dependencies (in the equi-correlated Gaussian one-sided model described in Example~\ref{sec:vectorpos}).
Obviously, a ``dynamic" choice of $\lambda$ (i.e., using the data) can increase the accuracy of the $(\pi_0(P))^{-1}$ estimation and thus should lead to a better procedure. However, proving that the corresponding FDR control remains valid in this case is an open issue to our knowledge. 
Also, outside the case of the particular equi-correlated Gaussian dependence structure, very little is known about adaptive FDR control. 

\begin{remark}\label{rem:onestage}
Some authors have proposed adaptive procedures that are not of the ``plug-in" form \eqref{equ_ASU}. For instance, 
we can define the class of ``one-stage step-up adaptive procedures", for which the threshold takes the form 
$t^{os}(\mbf{p})= \max \big\{ u \in [0,1] \telque \wh{\mathbb{G}}( \mbf{p},u) \geq r_\alpha(u) \big\},$
where $r_\alpha(\cdot)$ is a non-decreasing function that depends neither on $\mbf{p}$ nor on $\pi_0(P)$, see, e.g., \cite{Neu2008,FDR2009,BR2009}. As an illustration, Blanchard and Roquain (2009) have introduced the curve defined by $r_\alpha(t)=(1+m^{-1})\:t/(t+\alpha(1-\alpha))$ if $t\leq \alpha$ and $r_\alpha(t)=+\infty$ otherwise, see \cite{BR2009}. They have proved that the corresponding step-up procedure $R_{t^{os}}$ controls the FDR at level $\alpha$ in the independent model (by using the property of  Lemma~\ref{conditionDC} (i)). Furthermore,  Finner et al. (2009) have introduced the ``asymptotically optimal rejection curve" (AORC) defined by $r_\alpha(t)=t/(\alpha+t(1-\alpha))$, see \cite{FDR2009}. By contrast with the framework of the present paper, they considered the FDR control only in an asymptotic manner where the number $m$ of hypotheses tends to infinity.  
They have proved that the AORC enjoys the following (asymptotic) optimality property: while several adaptive procedures based on the AORC  provide a valid asymptotic FDR control (under independence), the AORC maximizes the asymptotic power among broad classes of adaptive procedures that asymptotically control the FDR, see Theorem~5.1,~5.3 and~5.5 in \cite{FDR2009}.
\end{remark}

\subsection{Case of arbitrary dependencies}

Many corrections of the linear step-up procedure are available to maintain the FDR control when the $p$-value family has arbitrary and unknown dependencies. We describe here the so-called ``Occam's hammer" approach presented in \cite{BF2007}. Surprisingly, it allows to recover and extend the well-known ``Benjamini-Yekutieli" correction \cite{BY2001} by only using Fubini's theorem. Let us consider 
\begin{align}
\label{equ_betaSU}
t^{\beta su}(\mbf{p})&=\max \{ u \in [0,1] \telque \wh{\mathbb{G}}( \mbf{p}, \beta(u)) \geq u/\alpha \}\\
 &=\max \{ u \in \{\alpha k /m ,1\leq k \leq m\} \telque \wh{\mathbb{G}}( \mbf{p}, \beta(u)) \geq u/\alpha \}\nonumber\\
&=\alpha /m \:\times \max \{  0\leq k \leq m \telque p_{(k)} \leq  \beta(\alpha k/m)\},\label{equ_betaSU-BHform}
\end{align}
for a non-decreasing function $\beta : \R^+ \rightarrow \R^+$. Then the FDR of $R_{\beta(t^{\beta su})}$ can be written as follows: for any $P\in \mathcal{P}$,
 \begin{align} 
\FDR(R_{\beta(t^{\beta su})},P)&= \alpha m^{-1} \sum_{i\in\cH_0(P)} \E\bigg[\frac{\ind{p_i\leq \beta(t^{\beta su}(\mbf{p}))}}{t^{\beta su}(\mbf{p})}\bigg]\nonumber\\
&= \alpha m^{-1} \sum_{i\in\cH_0(P)} \E\bigg[  \ind{p_i\leq \beta(t^{\beta su}(\mbf{p}))} \int_{0}^{+\infty} u^{-2}\ind{t^{\beta su}(\mbf{p})\leq u} du \bigg]\nonumber.
\end{align}
Next, using Fubini's theorem, we obtain
 \begin{align}
\FDR(R_{\beta(t^{\beta su})},P)&=\alpha m^{-1} \sum_{i\in\cH_0(P)} \int_{0}^{+\infty}u^{-2} \E\big[ \ind{t^{\beta su}(\mbf{p})\leq u}  \ind{p_i\leq \beta(t^{\beta su}(\mbf{p}))} \big] du\nonumber\\
&\leq  \alpha m^{-1} \sum_{i\in\cH_0(P)} \int_{0}^{+\infty} u^{-2}\P(p_i\leq \beta(u) ) du\nonumber\\
&=\alpha \frac{m_0(P)}{m} \int_{0}^{+\infty}u^{-2} \beta(u)  du\label{upbound_dep}.
\end{align}
Therefore, choosing any non-decreasing function $\beta$ such that $\int_{0}^{+\infty}u^{-2} \beta(u)  du=1$ provides a valid FDR control.
This leads to the following result:
\begin{theorem}
\label{thFDR-dep}
Consider a function $\beta : \R^+ \rightarrow \R^+$ of the following form: for all $u\geq 0$,
\begin{equation}\label{equbeta}
\beta(u)= \sum_{i:1\leq i \leq m, \alpha i/m\leq u} (\alpha i/m) \nu_i,
\end{equation}
where the $\nu_i$s are nonnegative with $\nu_1 + \dots + \nu_m = 1$.
Consider the step-up procedure $R_{\beta(t^{\beta su})}$ using $t^{\beta su}$ defined by \eqref{equ_betaSU}. Then for any $P\in \mathcal{P}$, $\FDR(R_{\beta(t^{\beta su})},P)\leq \alpha m_0(P)/m$.
\end{theorem}

Note that the function $\beta$ defined by \eqref{equbeta} takes the value $(\alpha/m) \nu_1+ \dots + (\alpha i/m) \nu_i$ in each $u=\alpha i /m$ and is constant on each interval  $(\alpha i /m,\alpha (i+1) /m)$ and on $(\alpha,\infty)$. Thus, it always satisfies that $\beta(u)\leq u$, for any $u\geq 0$. This means that the procedure $R_{\beta(t^{\beta su})}$ rejects always less hypotheses than the linear step-up procedure $R_{t^{su}}$. Therefore, while $R_{\beta(t^{\beta su})}$ provides a FDR control under no assumption about the $p$-value dependency structure, it is substantially more conservative than $R_{t^{su}}$ under weak PRDS dependencies between the $p$-values.

As an illustration, taking $\nu_i=i^{-1} \delta^{-1}$ for $\delta=1+1/2+...+1/m$, we obtain $\beta(\alpha i/m) = \delta^{-1}\alpha i/m $, which corresponds to the linear step-up procedure, except that the level $\alpha$ has been divided by $\delta\simeq \log(m)$. This is the so-called Benjamini-Yekutieli procedure proposed in \cite{BY2001}.
Theorem~\ref{thFDR-dep} thus recovers Theorem~1.3 of   \cite{BY2001}.
We mention another example, maybe less classical, to illustrate the flexibility of the  choice of $\beta$ in Theorem~\ref{thFDR-dep}. By taking $\nu_{m/2}=1$ and $\nu_i=0$ for $i\neq m/2$ (assuming that $m/2$ is an integer), we obtain $\beta(\alpha i/m) = (\alpha/2) \:\ind{i\geq m/2}$. In that case, the final procedure $R_{\beta(t^{\beta su})}$ rejects the hypotheses corresponding to $p$-values smaller than $\alpha/2$ if $2 p_{(m/2)}\leq \alpha$
and rejects no hypothesis otherwise. Theorem~\ref{thFDR-dep} ensures that this procedure also controls the FDR, under no assumption on the model dependency. 
Many other choices of $\beta$ are given in Section~4.2.1 of \cite{BR2008}.

Finally, let us underline that any FDR control valid under arbitrary dependency suffers from a lack of interpretability for the underlying FDP, as discussed in Section~\ref{sec:discussFDP}. 

\begin{remark}[Sharpness of the bound in Theorem~\ref{thFDR-dep}]\label{sec:sharpBY}
In Lemma 3.1 (ii) of \cite{LR2005} (see also \cite{GR2008}), a specifically crafted $p$-value distribution was built on $[0,1]^m$ (depending on $\beta$)
 for which the FDR of $R_{\beta(t^{\beta su})}$ is \textit{equal} to $\alpha $ (and $m_0(P)=m$). If the underlying model $\cP$ is such that $(p_i(X))_{1\leq i \leq m}$ can have this very specific distribution for some $P\in\cP$, the inequality ``$P\in \mathcal{P}$, $\FDR(R_{\beta(t^{\beta su})},P)\leq \alpha $" in Theorem~\ref{thFDR-dep} is sharp. 
 However, for a ``realistic" model $\cP$, this $p$-value distribution 
 is rarely attained because it assumes quite unrealistic dependencies between the $p$-values. 
 Related to that, several simulation experiments showed that the standard LSU procedure still provides a good FDR control under ``realistic" dependencies, see e.g. \cite{Far2006,KW2008}.
This means that the corrections defined in this section are generally very conservative for real-life data, because
 their actually achieved FDR  is much smaller than $\alpha m_0(P)/m$. Finally, another drawback of the bound of Theorem~\ref{thFDR-dep} is that it is much smaller than $\alpha$ when $\pi_0(P)=m_0(P)/m$ is small. To investigate this problem, we can think to apply techniques similar to  those of Section~\ref{sec:ALSU}. However, the problem of adaptive FDR control  is much more challenging under arbitrary dependency. The few  results that are available in this framework are very conservative, see \cite{BR2009}.   
 \end{remark}

\begin{remark}[Aggregation of dependent $p$-values]\label{rem:bul}
Consider Theorem~\ref{thFDR-dep} in the particular case where all $p$-values test the same null hypothesis, that is $\Theta_{0,i}=\Theta_0$ for any $i$. According to Remark~\ref{rem:allnullthesame}, we obtain a new test of level $\alpha$, by rejecting $H_0$: ``$P\in \Theta_0$" if the procedure $R_{\beta(t^{\beta su})}$ defined in Theorem~\ref{thFDR-dep} rejects at least one null hypothesis, that is, if there exists $k\geq 1$ such that  $p_{(k)} \leq  \beta(\alpha k/m)$.
As an illustration, taking $\nu_{\gamma m}=1$ and $\nu_i=0$ for $i\neq \gamma m$, for a given $\gamma\in[0,1]$ such that $\gamma m\in\{1,...,m\}$, we obtain
 $\beta(\alpha i/m) = (\alpha \gamma) \:\ind{i\geq \gamma m} $, which gives rise to a test rejecting $H_0$ whenever $p_{(\gamma m)} \gamma^{-1} \leq \alpha$. This defines a new global $p$-value $$\wt{p}=\min(p_{(\gamma m)} \gamma^{-1},1)$$ for testing $H_0$ 
  that can be seen as an aggregate of the original $p$-values. Thus, Theorem~\ref{thFDR-dep} shows that $\P(\wt{p}\leq \alpha)\leq\alpha$ under the null,  for arbitrary dependencies between the original $p$-values.
Interestingly, this aggregation procedure was independently discovered in \cite{MMB2009} in a context where one aims at combining $p$-values that were obtained by different splits of the original sample.
Also note that $\gamma=1/m$ corresponds to the Bonferroni aggregation procedure. Let us finally discuss the choice $\gamma = 1/2$ (assuming that $m/2$ is an integer). In that case, the aggregated $p$-value is $\wt{p}=\min(2\: p_{(m/2)} ,1)$. 
According to Remark~\ref{sec:sharpBY}, the factor  ``$2$" in the latter is needed in theory but may be over-estimated for a ``realistic" distribution of the $p$-value family. As a matter of fact, van de Wiel et al. (2009) have (theoretically) proved that this factor can be dropped as soon as  the $p$-value family has some underlying multivariate Gaussian dependency structure, see \cite{WBW2009}.
 \end{remark}

\section{$\kFWER$ control}\label{sec:kFWER}

The methodology presented in this section for controlling the $\kFWER$ under arbitrary dependencies can probably be attributed to many authors, e.g. \cite{Holm1979,WY1993,RW2005,RW2007}. Here, we opted for a general presentation which emphasizes the rationale of the mathematical argument.
This approach has been sketched in the talk \cite{Blanch2009} and investigated more deeply in  \cite{GS2010} where it is referred to as the ``sequential rejection principle". 
While the latter point of view allows to obtain elegant proofs, it is also useful for developing new $\FWER$ controlling procedures (e.g., hierarchical testing, Schaffer improvement), see \cite{GS2010,GF2010,KRW2010}.  
This methodology has been initially developed for the FWER. We propose in Section~\ref{sec:kFWERext} a new extension to the $\kFWER$.

In this section, for simplicity, we drop the explicit dependence of the multiple testing procedure $R$ w.r.t. $\mbf{p}$ in the notation.
The  parameter $k$ is fixed in $\{1,...,m\}$. 

\subsection{Subset-indexed family}

As a starting point, we assume that there exists a subset-indexed family $\{R_\cC\}_{\cC\subset \cH}$ of multiple testing procedures  satisfying the two following assumptions:
\begin{itemize}
\item[\textbullet] $\cC\mapsto R_\cC$ is non-increasing, that is, 
\begin{align}\label{NI:kFWER}\tag{\mbox{NI}}\forall \cC,\cC' \subset \cH\mbox{ such that }  \cC\subset \cC', \mbox{ we have }R_{\cC'}\subset R_{\cC};\end{align}
\item[\textbullet] $R_\cC$ controls the $\kFWER$ when $\cC$ is equal to the subset of true null hypotheses, 
that is,
\begin{align}\label{FWC0}\tag{$\mbox{FWC}_0$}\forall P \in \mathcal{P} \mbox{, }\kFWER(R_{\cH_0(P)},P)\leq \alpha.\end{align} 
\end{itemize}

A natural way of deriving such a family is to take a thresholding-based family of the form
\begin{equation}\label{threshold-fam}
R_\cC=\{1\leq i\leq m\telque p_i\leq t_{\cC}\},
\end{equation}
where $t_{\cC}\in[0,1]$ is a threshold which possibly depends on the data $\mathbf{p}=(p_i)_{1\leq i \leq m}$. Assumption \eqref{NI:kFWER} then holds as soon as we take $t_{\cC}$ non-increasing in $\cC$ (if $\cC\subset \cC'$ then $t_{\cC'} \leq t_{\cC}$). However, $t_{\cC}$ should be carefully chosen in order to ensure \eqref{FWC0}, as we discuss below.

A first instance of a thresholding-based family satisfying \eqref{NI:kFWER}-\eqref{FWC0} is the ``Bonferroni family" that chooses $t_{\cC}=\min(\alpha k/|\cC|,1)$. Condition \eqref{FWC0} results from Markov's inequality:
$$
 \P ( |\cH_0(P)\cap R_{\cH_0(P)}|\geq k) \leq k^{-1} \sum_{i\in\cH_0(P)}  \P(p_i\leq t_{\cH_0(P)} ) \leq  |\cH_0(P)| t_{\cH_0(P)}/k \leq \alpha.
$$
This family is not adaptive w.r.t. the dependence structure of the $p$-values. As an illustration, when the true $p$-values are all equal, say, to $p_{i_0}$, $i_0\in\cH_0(P)$, we have
$$
 \P ( |\cH_0(P)\cap R_{\cH_0(P)}|\geq k)= \P ( |\cH_0(P)|\ind{p_{i_0}\leq t_{\cH_0(P)}} \geq k) \leq   t_{\cH_0(P)}.
 $$
Thus, under this extreme dependency structure, the Bonferroni threshold $\min(\alpha k/|\cC|,1)$ can be replaced by $\alpha$ (the only case which matters is $|\cC|\geq k$, see Remark~\ref{rem-augmentation} below). Hence, there is a potential loss when using the Bonferroni family.
In practice, the Bonferroni family is often used as a ``benchmark family" for evaluating the performance of other families.

In order to improve on the Bonferroni family, one can try to choose a threshold $t_{\cC}$ that captures the dependencies between the $p$-values while  still satisfying \eqref{NI:kFWER}-\eqref{FWC0}. 
For this, first note that for $R_\cC$ defined by \eqref{threshold-fam},
\begin{align*}
\kFWER(R_{\cC},P) &= \P (\exists i_1,...,i_k\in \cH_0(P) \telque \forall i\in \{i_1,...,i_k\}, p_i\leq t_\cC )\\
&= \P (\mbox{k-min}\{p_i,i\in \cH_0(P)\} \leq t_\cC), 
\end{align*}
where $\mbox{k-min}\{p_i,i\in\cH_0(P)\}$ denotes the $k$-th smallest element of $\{p_i,i\in\cH_0(P)\}$. Therefore, 
 a natural choice for $t_\cC$ is the $\alpha$-quantile of the distribution of $\mbox{k-min}\{p_i,i\in \cC\}$. However, the latter is generally unknown because the underlying distribution $P$ is unknown. An idea is to approximate it by using a randomized thresholding procedure. 
This method can be applied when the null hypothesis is invariant under the action of a finite group of transformations of the original observation set $\cX$ onto itself (such a transformation can be for instance a permutation or a sign-flipping, see \citep{RW2005,RW2007,ABR2010a,ABR2010b}). 
For a recent and general description of this method, we refer the reader to Theorem~2 of \cite{GS2010} (while \cite{GS2010} have developed this method only for $k=1$, it can be directly generalized to the case of $k\geq 1$). 
The resulting family satisfies \eqref{NI:kFWER}-\eqref{FWC0} 
 while it is ``adaptive" with respect to the $p$-value dependence structure, in the sense that $t_\cC=t_\cC(\mbf{p})$ implicitly takes into account the potential relations existing between the $p$-values.


\begin{remark}\label{rem-GS}
The monotonicity condition introduced in \cite{GS2010} can be rewritten with our notation as follows:
\begin{align}\label{LNI:kFWER}\tag{\mbox{wNI}}\forall \cC,\cC' \subset \cH\mbox{ such that }  \cC\subset \cC', \mbox{ we have }R_{\cC'}\cap \cC'\subset R_{\cC}.\end{align}
Condition \eqref{LNI:kFWER} is weaker than condition \eqref{NI:kFWER}. Thus, at first sight, the setting of  \cite{GS2010} is more general than ours. The next reasoning shows that the two settings are in fact equivalent. Since the condition \eqref{FWC0} only depends on the set of $R_\cC\cap \cC$ (for $\cC=\cH_0$), we can add the elements of $\cC^c$ in the rejection set $R_\cC$ while still maintaining \eqref{FWC0} true. Therefore, starting from a subset-indexed family $\{R_\cC\}_{\cC\subset \cH}$ satisfying the weaker assumptions \eqref{LNI:kFWER}-\eqref{FWC0}, we may define a new subset-indexed family $\{R'_\cC\}_{\cC\subset \cH}$ satisfying our assumptions \eqref{NI:kFWER}-\eqref{FWC0}, by letting $R'_\cC=R_\cC\cup \cC^c$, and then apply to this family the methodology described in the next sections. Moreover, by anticipating the definition of the FWER-controlling algorithm that will be presented
in  Section~\ref{sec:kFWERext}, we can easily check that the output of this algorithm applied to the family $\{R'_\cC\}_{\cC\subset \cH}$ is the same than the algorithm of \cite{GS2010} applied to the family $\{R_\cC\}_{\cC\subset \cH}$.
As a consequence, our framework covers the original setting of \cite{GS2010}. 
\end{remark}

\begin{remark}\label{rem-augmentation}
Any subset-indexed family $\{R_\cC\}_{\cC\subset \cH}$ satisfying \eqref{NI:kFWER}-\eqref{FWC0} can be modified in the following way: take $\wt{R}_\cC=\cH$ (reject all hypotheses) when $|\cC|<k$ and $\wt{R}_\cC={R}_\cC$ otherwise.  This maintains the conditions \eqref{NI:kFWER}-\eqref{FWC0}, because the $\kFWER$ is always zero when $|\cH_0(P)|<k$.
\end{remark}

In what follows, we investigate the problem of the $\kFWER$ control once we have fixed a subset-indexed family $\{R_\cC\}_{\cC\subset \cH}$ satisfying \eqref{NI:kFWER}-\eqref{FWC0}.

\subsection{Single-step method}\label{sec:kFSER-ss}

From assumption \eqref{FWC0}, the procedure $R_{\cH_0(P)}$ using $\cC=\cH_0(P)$ controls the $\kFWER$. Clearly, this procedure cannot be used because $\cH_0(P)$ depends on the unknown  underlying distribution $P$ of the data. 
We can use instead $R_{\cC}$ with $\cC=\cH$ because, from the two assumptions \eqref{NI:kFWER}-\eqref{FWC0} above, we have $\kFWER(R_{\cH},P)\leq \kFWER(R_{\cH_0(P)},P)\leq \alpha$. This implies that $R_{\cH}$ always controls the $\kFWER$ at level $\alpha$. The latter is generally called the \textit{single-step} procedure (associated to the family $\{R_\cC\}_{\cC\subset \cH}$). 
However, we argue that $R_{\cH}$ could be often too conservative w.r.t. $R_{\cH_0(P)}$, for the two following reasons:
\begin{itemize}
\item $\cH_0(P)$ can be much smaller than $\cH$;
\item the way the procedures $\{R_\cC\}$ have been built implicitly assumed that $\cC=\cH_0(P)$ and can be very conservative when $\cC$ is much larger than $\cH_0$.   
\end{itemize}
For instance, these behaviors have been extensively discussed in \cite{ABR2010b} 
for particular Rademacher-resampled thresholding procedures. 
Therefore, we seek for a procedure controlling the $\kFWER$ which is ``close" to $R_{\cH_0(P)}$ and which can be derived from the family $\{R_{{\cC}}\}_{\cC\subset \cH}$ via a simple algorithm. 

\subsection{Step-down method for FWER}

We present in this section the special case of $k=1$, following the approach of  \cite{RW2005} with the presentation proposed in \cite{Blanch2009,GS2010}. 
Let us denote by $A_\cC$ the sets $(R_\cC)^c$ of non-rejected hypotheses for the subset-indexed family. Consider the event $$\Omega_0 = \{ R_{\cH_0(P)}\cap \cH_0(P)= \emptyset\}= \{ \cH_0(P) \subset A_{\cH_0(P)}\}.$$
 By assumption \eqref{FWC0}, we have $\P(\Omega_0)\geq 1-\alpha$. 
Since from \eqref{NI:kFWER}, $A_{\cC}$ is non-decreasing in $\cC$, the following holds on $\Omega_0$: for any $\cC\subset \cH$, 
 \begin{equation}\label{transform-FWER}
\cH_0(P) \subset \cC \Longrightarrow A_{\cH_0(P)}\subset A_\cC  \Longrightarrow  \cH_0(P)  \subset A_{\cC}.
 \end{equation}
 Thus, on the event $\Omega_0$, taking $\cC=\cC_0=\cH$ in \eqref{transform-FWER} gives that  $\cH_0(P) \subset A_{\cC_0}$, which in turn implies $\cH_0(P)  \subset A_{\cC_1}$ by taking $\cC=\cC_1=A_{\cC_0}$ in  \eqref{transform-FWER}, and so on. By recursion, this proves the following result:
 
 \begin{theorem}\label{th-FWER}
Assume that a family $\{R_\cC\}_{\cC\subset \cH}$ of multiple testing procedures satisfies conditions \eqref{NI:kFWER} and \eqref{FWC0} and consider the corresponding family of non-rejected hypotheses $\{A_\cC\}_{\cC\subset \cH}$. Define $\hat{\cC}$ by the following ``step-down" recursion:
\begin{itemize}
\item[\textbullet] Initialization: $\cC_0= \cH$; 
\item[\textbullet] Step $j\geq 1$: let $\cC_j=A_{\cC_{j-1}}$. If $\cC_{j}=\cC_{j-1}$, let $\hat{\cC}=\cC_j$ and stop. Otherwise go to step $j+1$;
 \end{itemize}
 Then the procedure $R=(\hat{\cC})^c$, which also equals $R_{\hat{\cC}}$, controls the FWER at level $\alpha$ for any $P\in\cP$.
\end{theorem}

Note that for all $j\geq 0$, we have $\cC_{j+1} \subset \cC_j $, because $\cC_1\subset \cC_0$ and $A_\cC$ is non-decreasing in $\cC$. Thus, the set of rejected hypotheses can only increase during the step-down algorithm.
In particular, 
the final procedure ${\hat{\cC}}^c=R_{\hat{\cC}}$ is always less conservative than the single-step procedure $R_\cH$, for the same FWER control. Thus, using a step-down algorithm is always more powerful than the single-step method.

\begin{example}[Bonferroni step-down procedure for $\FWER$ control]
Theorem~\ref{th-FWER} can be used with the Bonferroni family $R_\cC=\{1\leq i \leq m\telque p_i\leq \alpha /|\cC|\}$. In that case, 
by reordering the $p$-values $p_{(1)}\leq ...\leq p_{(m)}$ (with $p_{(0)}=0$),  the corresponding step-down procedure defined in Theorem~\ref{th-FWER} can be reformulated as 
rejecting the nulls with $p_i\leq \alpha /(m-\hat{\ell}+1)$, where $\hat{\ell}=\max\{\ell\in\{0,1,...,m\}\telque \forall \ell'\leq \ell , \:p_{(\ell')}\leq \alpha /(m-{\ell'}+1) \}$.
This is the well known step-down \textit{Holm procedure} which was introduced and proved to control the FWER in \cite{Holm1979}. By contrast with step-up procedures, the step-down Holm procedure starts from the most significant $p$-value and stops the first time that a (ordered) $p$-value exceeds the critical curve. This is illustrated in Figure~\ref{fig_stepdown}. 
\end{example}

 \begin{figure}[htbp]
\begin{center}
\includegraphics[scale=0.35,angle=-90]{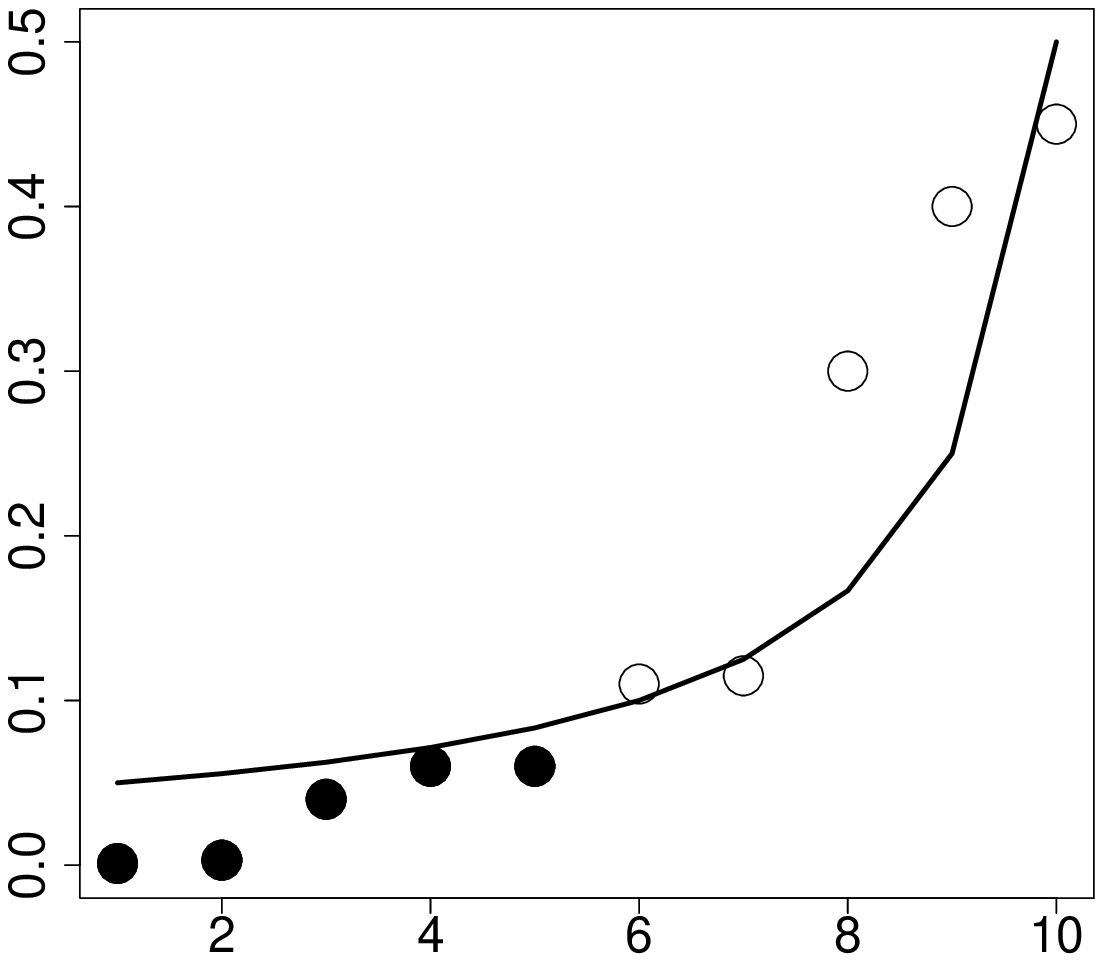}
\includegraphics[scale=0.35,angle=-90]{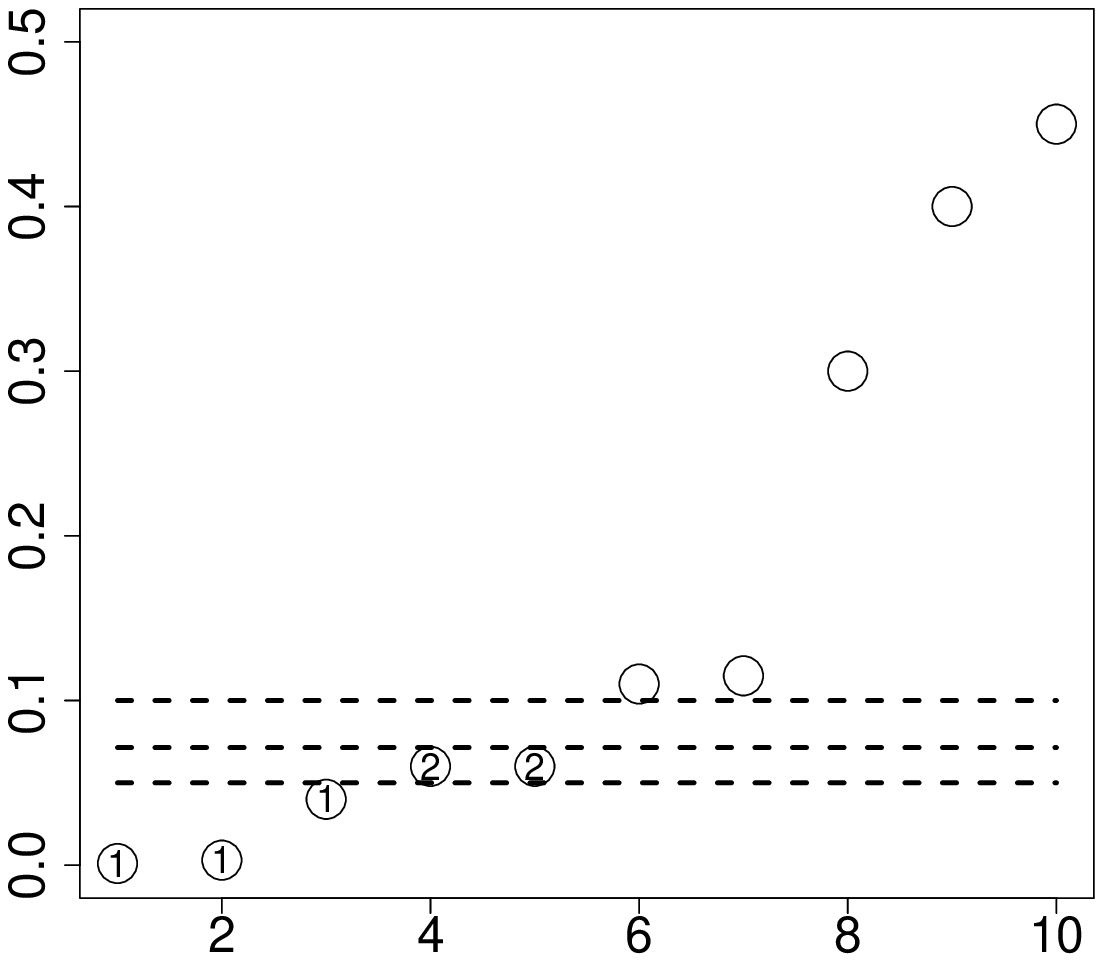}
\begin{pspicture}(0,0)(0,0)
\rput[l](-0.3,-5.3){{\tiny $\alpha/10$}}
\rput[l](-0.3,-5.05){{\tiny $\alpha/7$}}
\rput[l](-0.3,-4.8){{\tiny $\alpha/5$}}
\end{pspicture}
\caption{Illustration of the two equivalent definitions of Holm's procedure.
The left picture is the classical step-down representation: ordered $p$-values together with the solid curve $\l\mapsto \alpha/(m-\l+1)$. The filled points represent $p$-values that corresponds to the rejected hypotheses.
The right picture illustrates the algorithm of Theorem~\ref{th-FWER}: ordered $p$-values with the three thresholds $\alpha/10$ (step 1), $\alpha/7$ (step 2) and $\alpha/5$ (step 3). 
For $i\in\{1,2\}$, the points filled with ``i" are rejected in the $i$th step of the algorithm.
Both pictures use the same $p$-values and $m=10$; $\alpha=0.5$.}
\label{fig_stepdown}
\end{center}
\end{figure}

\subsection{Step-down method for $\kFWER$}\label{sec:kFWERext}

We would like to generalize Theorem~\ref{th-FWER} to the case of the $\kFWER$. This time, we should consider the event 
 \begin{align*}
 \Omega_0 &= \{ |R_{\cH_0(P)}\cap \cH_0(P)| \leq k-1\} 
= \{ \exists I_0\subset \cH, |I_0|= k-1 : \cH_0(P) \subset A_{\cH_0(P)}\cup I_0 \},
 \end{align*}
which satisfies by assumption $\P(\Omega_0)\geq 1-\alpha$. For any subset $\cC\subset \cH$, let
\begin{equation}\label{equ-pfkFWER}
\phi(\cC)=\bigcup_{I\subset \cH, |I|= k-1} A_{\cC\cup I}=\bigcup_{I\subset \cC^c, |I|\leq k-1} A_{\cC\cup I}
.  
\end{equation}
Then we may prove that the following holds: on the event $\Omega_0$, for any $\cC\subset \cH$,
 \begin{align*}
  \exists I\subset \cH, |I| = k-1 : \cH_0(P) \subset \cC\cup I &\: \Longrightarrow\: \exists I\subset \cH, |I| = k-1 : A_{\cH_0(P)} \subset A_{\cC\cup I} \subset \phi(\cC)\\
 &\:\Longrightarrow\: \exists I'\subset \cH, |I'|= k-1 : \cH_0(P) \subset \phi(\cC) \cup I'.
 \end{align*}
 The first implication holds because $A_{\cC}$ is non-decreasing in $\cC$ and the second implication holds by considering $I'=I_0$.
Thus, on the event $\Omega_0$, for any $\cC\subset \cH$,
$$|\cC^c\cap \cH_0(P)| \leq k-1 \Longrightarrow |(\phi(\cC))^c \cap \cH_0(P)| \leq k-1.$$
This leads to the following result.

\begin{theorem}\label{th-kFWER}
Assume that a family $\{R_\cC\}_{\cC\subset \cH}$ of multiple testing procedures satisfies conditions \eqref{NI:kFWER} and \eqref{FWC0} and consider the corresponding family of non-rejected hypotheses $\{A_\cC\}_{\cC\subset \cH}$ and let $\phi$ be defined by \eqref{equ-pfkFWER}. Define $\hat{\cC}$ by the following ``step-down" recursion:
\begin{itemize}
\item[\textbullet] Initialization: $\cC_0= \cH$; 
\item[\textbullet] Step $j\geq 1$: let $\cC_j=\phi(\cC_{j-1})$. If $\cC_{j}=\cC_{j-1}$, let $\hat{\cC}=\cC_j$ and stop. Otherwise go to step $j+1$;
 \end{itemize}
 Then the procedure $R = (\hat{\cC})^c$, which also equals $(\phi(\hat{\cC}))^c=\bigcap_{|I|= k-1} R_{\hat{\cC}\cup I}$, controls the $\kFWER$ at level $\alpha$ for any $P\in\cP$.
\end{theorem}

From \eqref{equ-pfkFWER}, $\phi(\cdot)$ is non-decreasing, that is, $\forall \cC\subset \cC'$, $\phi(\cC)\leq \phi(\cC')$. 
As a consequence, we derive from $\cC_1\subset \cC_0$ that $\cC_{j} \subset \cC_{j-1}$ for all $j\geq 1$. Therefore, the rejection set can only increase at each step of the step-down algorithm. 
In particular, 
the final procedure ${\hat{\cC}}^c=\bigcap_{|I|= k-1} R_{\hat{\cC}\cup I}$ is always less conservative than the single step method $R_\cH$, for the same $\kFWER$ control. Therefore, using the step-down algorithm always leads to a power improvement.

To illustrate Theorem~\ref{th-kFWER}, let us consider a thresholding-based family of the form $R_\cC=\{1\leq i\leq m\telque p_i\leq t_{\cC}\}$ with a non-increasing threshold function $\cC\mapsto t_{\cC}$ (i.e., such that for $\cC\subset \cC'$, we have $t_{\cC'}\leq t_{\cC}$) and such that $\{R_\cC\}_{\cC}$ satisfies \eqref{FWC0}. The recursion relation $\cC'=\phi(\cC)$ can be rewritten in that case as follows:
\begin{align*}
(\cC')^c &= \bigcap_{I\subset \cC^c, |I|\leq k-1} R_{\cC\cup I}\\
&= \bigcap_{I\subset \cC^c, |I|\leq k-1} \{1\leq i \leq m\telque p_i\leq t_{\cC\cup I} \}\\
&=\big\{1\leq i \leq m\telque p_i\leq \min_{I\subset \cC^c, |I|\leq k-1}\{ t_{\cC\cup I}\} \big\}.
\end{align*}
This recovers the generic step-down method described in Algorithm~2.1 of \cite{RW2007}, which was developed in the case where the subset-indexed family is thresholding based. 

\begin{example}[Bonferroni step-down procedure for $\kFWER$ control]
When we choose the Bonferroni family, i.e., the threshold family $t_\cC=\alpha k/|\cC|$, we have 
$$\min_{I\subset \cC^c, |I|\leq k-1}\{ t_{\cC\cup I}\}= \frac{\alpha k}{m\wedge (|\cC| + k-1)}.$$
Therefore, in terms of the ordered $p$-values $0=p_{(0)}\leq p_{(1)}\leq ...\leq p_{(m)}$, the procedure of Theorem~\ref{th-kFWER} can be reformulated as rejecting the null $H_{0,i}$ when $p_i\leq \alpha k/(m\wedge(m-\hat{\ell}+k))$ where $\hat{\ell}=\max\{\ell\in\{0,1,...,m\}\telque \forall \ell'\leq \ell , \:p_{(\ell')}\leq \alpha k/(m\wedge(m-{\ell'}+k)) \}$.
The latter is the \textit{generalized Holm procedure}, which was introduced and proved to control the $\kFWER$ in \cite{LR2005}. 
\end{example}

\section{FDP control}\label{sec:FDP}

The problem of controlling the FDP has been investigated in many studies, e.g., \cite{LR2005,LBH2005,RS2006b,CT2008,RW2007,DL2008,RW2010}.
We follow here a methodology proposed by Romano and Wolf (2007), see \cite{RW2007}. They have proposed  to use a family $\{S_k\}_k$ of $\kFWER$ controlling procedures and to choose $k$ that ensures that the corresponding rejection number $|S_k|$ is ``sufficiently large". Roughly speaking, choosing $k$ such that $|S_k|$ is larger than $(k-1)/\gamma$ implies that, with high probability,
  $$\FDP(S_k,P)= |S_k\cap \cH_0(P)|/|S_k| \leq (k-1)/|S_k| \leq \gamma.$$
Obviously, as it is, the above reasoning is not rigorous, because the chosen $k$ depends on the data. Theorem~4.1 (i) of \cite{RW2007} establishes that the latter approach leads to a correct FDP control in the asymptotic setting where the sample size available for each test tends to infinity. This can be seen as a Dirac configuration where each $p$-value corresponding to false nulls are equal to zero. 

In this section, we propose to reformulate this approach by using as index the rejection number instead of $k$. 
Roughly speaking, if we choose $\{R_\l\}_\l$ such that each $R_\l$ controls the $(\gamma\l +1)$-FWER and we choose $\l$ such that $|R_\l|\geq \l$, we obtain that, with high probability,
  $$\FDP(R_\l,P)= |R_\l\cap \cH_0(P)|/|R_\l| \leq \gamma\l/|R_\l| \leq\gamma.$$
Similarly to the previous paragraph, this argument is not rigorous because the chosen $\l$ depends of the data. The main task of this section is to rationalize this approach.
This leads to a general result (Theorem~\ref{FDPfromkFWER} given in Section~\ref{sec:FDPfromkFWER}), which covers both Theorem~4.1 (i) of \cite{RW2007} in the ``Dirac" setting (see Section~\ref{sec:FDPcor2}) and the earlier result of \cite{LR2005} (see Section~\ref{sec:FDPcor1}). 
As additional corollary, we derive the FDP control of the quantile-binomial procedure described in Algorithm~\ref{equ-quantilebinom}, when the data are assumed to follow the model $\cP^I$ (see Section~\ref{sec:FDPcor1}). 

In this section, the  parameter $\gamma$ is fixed once and for all in $(0,1)$.

\subsection{Family indexed by rejection numbers}\label{sec:volindex}

Assume that we have at hand a family $\{R_\l\}_{1\leq \l \leq m}$ of multiple testing procedures and a class of distributions $\mathcal{P}'\subset \cP$ satisfying the following properties:
\begin{itemize}
\item[\textbullet] $R_\l$  is non-decreasing with respect to $\l$, that is,
\begin{align}
\forall \l\in\{1,...,m-1\},\:\: R_\l\subset R_{\l+1}\,;\label{ND:FDP}\tag{\mbox{ND}}
\end{align}
\item[\textbullet] $R_\l$ controls the $\mbox{$(\lfloor\gamma\l\rfloor +1)$-\FWER}$ at level $\alpha$ for any $P \in \cP'$ such that less than $m - \l + \lfloor\gamma(\l-1)\rfloor +1$ null hypotheses are true, that is,
\begin{align}
\begin{array}{c}\forall \l\in\{1,...,m\},\:\:\: \forall P \in \cP' \mbox{ s.t. } |\cH_0(P)|\leq m - \l + \lfloor\gamma(\l-1)\rfloor +1, \\
\P(|R_\l\cap \cH_0(P)|\geq \lfloor\gamma\l\rfloor +1)\leq \alpha\,\end{array};\label{FWC}\tag{\mbox{FWC}}
\end{align}
\item[\textbullet] for any $P\in \cP'$, for any $\l\in\{1,...,m\}$, the false rejection number of $R_\l$ is independent of the correct rejection numbers of $R_{\l'}$, for $1\leq \l' \leq m$, that is,
\begin{align}
\forall P\in \cP', \forall \l\in\{1,...,m\}, |R_\l\cap \cH_0(P)| \mbox{ is independent of } \{|R_{\l'}\cap \cH_1(P)|, 1\leq \l' \leq m\}\,.\label{DA}\tag{\mbox{DA}}
\end{align}
\end{itemize}

In condition \eqref{FWC}, for any $x\geq 0$, $\lfloor x \rfloor$ denotes the largest integer $n$ such that $n\leq x$.
Condition \eqref{ND:FDP} is natural because the index $\l$ can be interpreted as a rejection number. It is easy to check in the examples below.

For any $\cP'\subset\cP$, condition \eqref{FWC} is fulfilled by the (single-step or step-down) $\kFWER$ controlling procedures of the previous section when $k=\lfloor\gamma\l\rfloor +1$. As a first instance, we can use the (single-step) Bonferroni family $R_\l$ using the threshold $\alpha (\lfloor\gamma\l\rfloor +1)/m$. Moreover, note that $|\cH_0(P)|\leq m - \l + \lfloor\gamma(\l-1)\rfloor +1$ in \eqref{FWC}, thus we can consider the improved threshold 
\begin{equation}\label{equ_seuilRom}
t^{LR}_\l = \frac{\alpha (\lfloor\gamma\l\rfloor +1)}{m - \l + \lfloor\gamma(\l-1)\rfloor +1}.
\end{equation}
The threshold \eqref{equ_seuilRom} is slightly larger 
than the threshold used in Theorem~3.1 of \cite{LR2005}
 (they used $\lfloor\gamma\l\rfloor$ instead of $\lfloor\gamma(\l-1)\rfloor$ in the denominator).
  As a second instance, we can substantially improve on the above threshold family when we additionally assume that the distribution $P$ of the data lies in the smaller subset $\cP'=\cP^I$:
for this, note that for any $P\in \cP^I$ and for any $t\in[0,1]$, the variable $|\{i\in \cH_0(P) \telque p_i(X)\leq  t\}|$ is stochastically upper-bounded by a binomial distribution of parameters $|\cH_0(P)|$ and $t$, which in turn is stochastically upper-bounded by a binomial distribution of parameters $m - \l + \lfloor\gamma(\l-1)\rfloor +1$ and $t$. Therefore, choosing the (deterministic) quantile-based threshold family $(t^Q_\l)_{1\leq \l\leq m}$ defined by
\begin{align}\label{equ_TLR_improved}
 t^Q_{\l}&=\max\{t\in[0,1] \telque \P\big( Z > \gamma {\l} \big) \leq \alpha \mbox{ for $Z\sim \mathcal{B}(m -\l+ \lfloor\gamma ({\l}-1)\rfloor + 1 , t)$}\}\\
 &=\max\{t\in[0,1] \telque q_\l(t)\leq \gamma \l \},\nonumber
\end{align}
where $q_{\l}(\cdot)$ is defined by \eqref{equ-quantilebinom},
we obtain a family of thresholding procedures satisfying \eqref{FWC} with $\cP'=\cP^I$.
Clearly, since $t^{LR}_\l$ in \eqref{equ_seuilRom} is only based upon Markov's inequality, which is in general not accurate for binomial variables, the threshold family  $ t^Q_{\l}$ defined by \eqref{equ_TLR_improved} is substantially larger, 
as illustrated in Figure~\ref{fig_LR}. Interestingly, we can use more elaborate deviation inequalities to obtain thresholds that are better than $t^{LR}_\l$ while having a form more explicit than $ t^Q_{\l}$, see Remark~\ref{rem:devbinom}.

Assumption \eqref{DA} is a dependence assumption which is typically satisfied in the two following cases:
\begin{itemize}
\item[$-$] each procedure $R_\l$ uses a deterministic threshold and the $p$-values associated to true nulls are independent of the $p$-values associated to false nulls, for all distributions of $\cP'$, that is,
\begin{align}
\begin{array}{c}\forall  \l\in\{1,...,m\}, R_\l = \{i \in \{1,...,m\} \telque  p_i \leq t_\l\} \mbox{ for a deterministic $t_\l\in[0,1]$} \\
\mbox{and }\forall P\in\mathcal{P}', (p_i(X))_{i\in\cH_0(P)} \mbox{ is independent of } (p_i(X))_{i\in\cH_1(P)}\,\end{array};\label{DA2}\tag{\mbox{DA'}}
\end{align}
\item[$-$] for all distributions of $\cP'$, the number of correct rejections of each $R_\l$ is deterministic, that is,
\begin{align}
\forall P\in \cP', \:\{|R_{\l'}\cap \cH_1(P)|, 1\leq \l' \leq m\} \mbox{ is deterministic}. \label{DA3}\tag{\mbox{DA''}}
\end{align}
\end{itemize}

Condition \eqref{DA3} is satisfied for instance when $\cH_1(P)\subset R_{\l'}$, for any $\l'$, which is the case for procedures of the form $R_\l = \{i \in \{1,...,m\} \telque  p_i \leq t_\l(\mbf{p})\}$ using a possibly data-dependent threshold $t_\l(\mbf{p})\in[0,1]$, when we assume that the $p$-values are in the Dirac configuration, that is, when they are equal to zero under the alternative.

\begin{remark}\label{rem:devbinom}
Using Hoeffding's and Bennett's inequalities (see, e.g.,  Proposition~2.7 and~2.8 in \cite{Mass2007}), we can derive a  family of thresholding procedures satisfying \eqref{FWC} with $\cP'=\cP^I$, by using
 the threshold 
\begin{equation}\label{equ_TLR_Cher}
(t^{Q})'_{\l}=\max(t^{LR}_\l,t^{Ho}_\l,t^{Be}_\l ),
\end{equation}
where we let
\begin{align*} 
t^{Ho}_\l&=\bigg(\frac{\lfloor\gamma\l\rfloor +1}{m - \l + \lfloor\gamma(\l-1)\rfloor +1}- \left(\frac{\log(1/\alpha)}{2(m - \l + \lfloor\gamma(\l-1)\rfloor +1)}\right)^{1/2}\bigg)\vee 0\\
t^{Be}_\l&=\frac{\lfloor\gamma\l\rfloor +1}{m - \l + \lfloor\gamma(\l-1)\rfloor +1} \:\:h^{-1}\left(\frac{\log(1/\alpha)}{\lfloor\gamma\l\rfloor +1}\right),
\end{align*}
 with $h(u)=u-\log(u)-1$, $u\in(0,1]$. 
\end{remark}

 \begin{figure}[htbp]
\begin{center}
\includegraphics[scale=0.3,angle=-90]{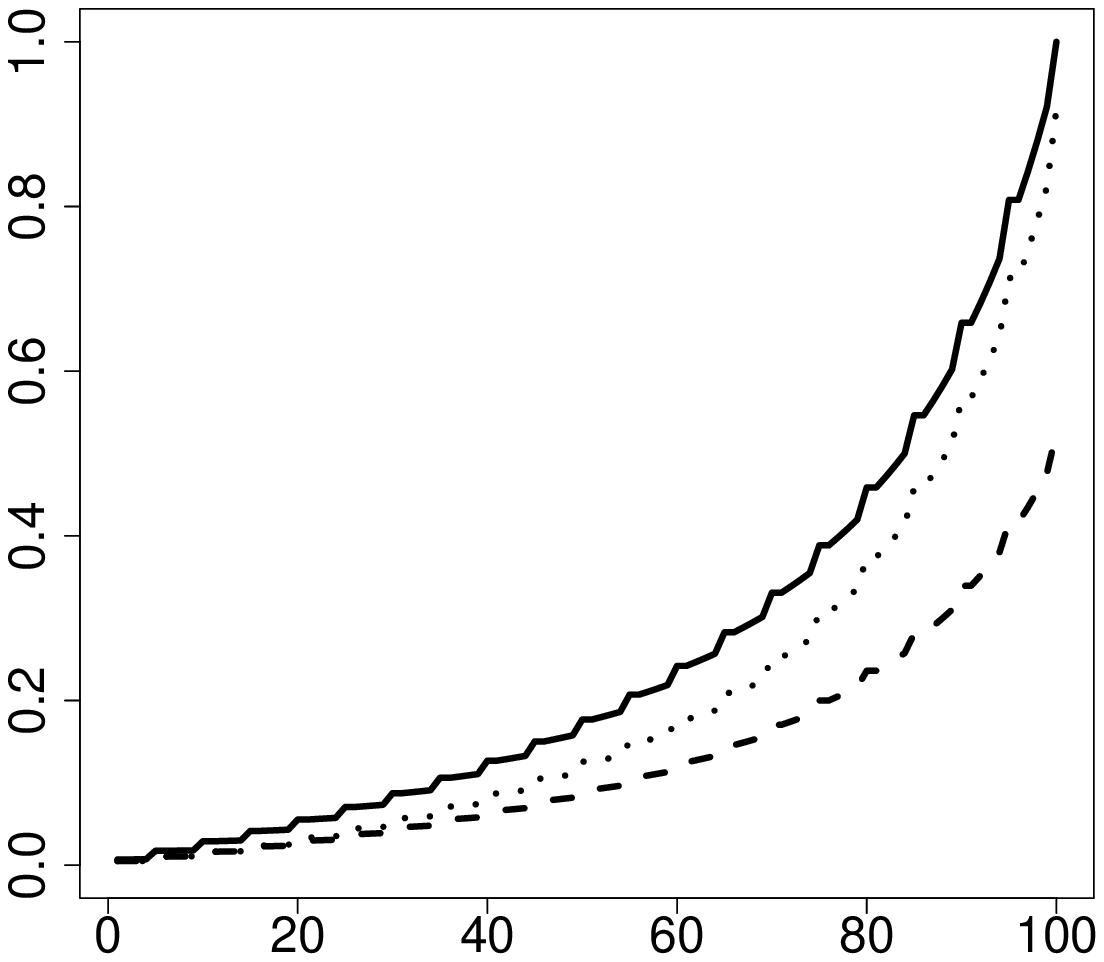}
\includegraphics[scale=0.3,angle=-90]{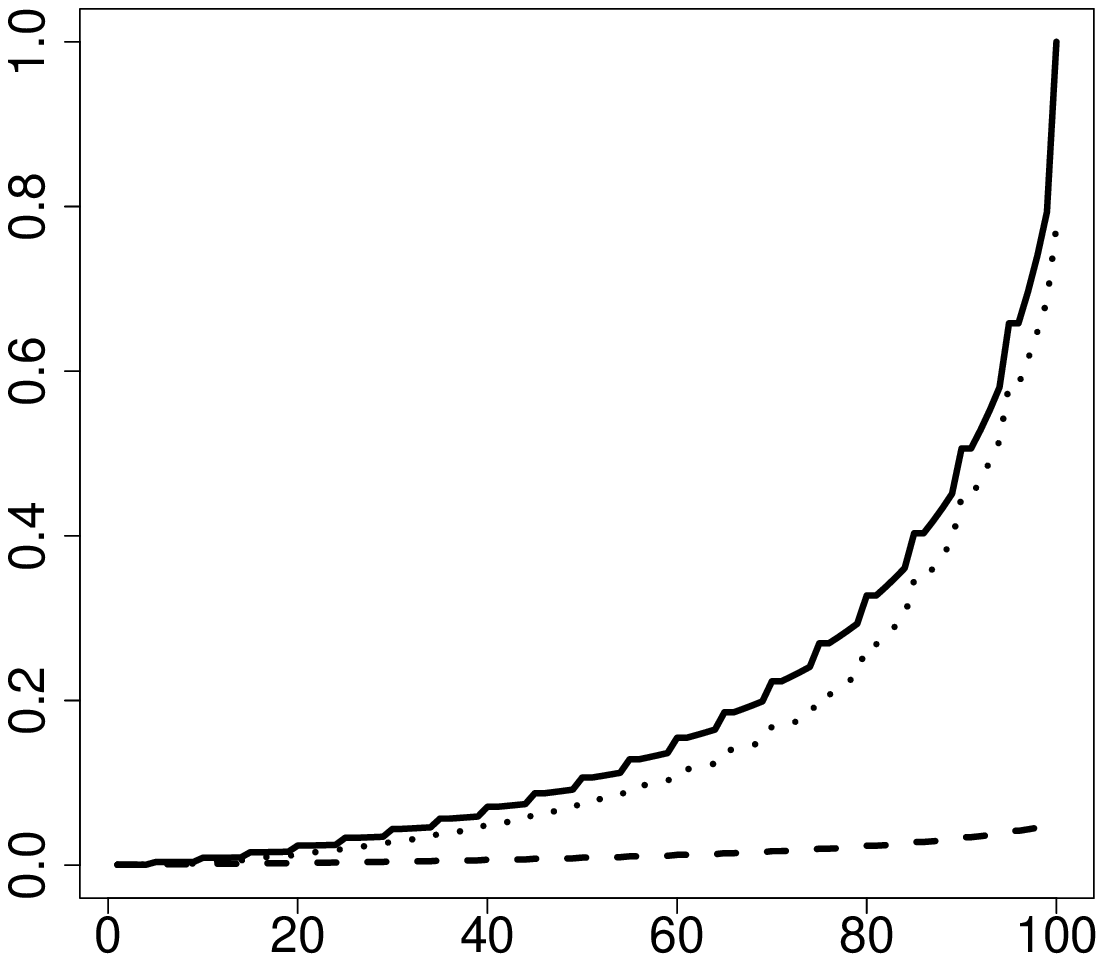}
\caption{Threshold  $t^Q_{\l}$ in \eqref{equ_TLR_improved} for model $\cP^I$ (solid line), threshold  $(t^Q)'_{\l}$ in \eqref{equ_TLR_Cher} for model $\cP^I$ (dotted line) and  threshold $t^{LR}_\l$ in  \eqref{equ_seuilRom} for model $\cP$ (dashed line) in function of $\l\in\{1,...,m\}$. $m=100$;  $\gamma=0.2$. Right: $\alpha=0.5$; left: $\alpha=0.05$.}
\label{fig_LR}
\end{center}
\end{figure}

\subsection{Step-down method}\label{sec:FDPfromkFWER}

The approach described in this section is an adaptation of the proof of Theorem~3.1 in \cite{LR2005} to our setting.
Let us consider a family $\{R_\l\}_{1\leq \l \leq m}$ and a class of distributions $\mathcal{P}'\subset \cP$ satisfying \eqref{ND:FDP}-\eqref{FWC}-\eqref{DA}. We aim at selecting $\l=\hat{\l}$ that provides $\forall  P \in\cP'$, $\FDP(R_{\hat{\l}},P)\leq \alpha$.

First note that, by definition of the FDP, we have for any $\l\in\{1,...,m\}$ such that $|R_\l|=\l$:
\begin{align}
\{\FDP(R_\l,P)>\gamma\} &= \{|\cH_0(P)\cap R_\l|> \gamma \l\} \nonumber\\
&=\{|\cH_0(P)\cap R_\l| \geq  \lfloor\gamma\l\rfloor +1\} \nonumber\\
&=\{ \l \in \mathcal{L}\},\label{rel-FDPfond}
\end{align}
where $\mathcal{L}=\{\l\in\{1,...,m\}\telque \l - |\cH_1(P)\cap R_\l|  \geq  \lfloor\gamma\l\rfloor +1 \}$ is a set which only depends on the set $\{|\cH_1(P)\cap R_{\l'} |, 1\leq \l' \leq m\}$.

Second, note that for any $\l\in\{1,...,m\}$ such that $|R_\l|  \geq \l$, 
\begin{align}\label{rel-FDPfond2}\{\l \in\mathcal{L}\}\subset \{|\cH_0(P)\cap R_\l| \geq  \lfloor\gamma\l\rfloor +1\}.\end{align}
Let us consider $\l^\star = \min\{ \mathcal{L}\}$ (with $\l^\star=m+1$ when $\mathcal{L}=\emptyset$).
From \eqref{rel-FDPfond} and \eqref{rel-FDPfond2},  taking  $\hat{\l}\in\{1,...,m\}$ such that $|R_{\hat{\l}}|=\hat{\l}$ and such that for any $\l\leq\hat{\l}$, $|R_\l|  \geq \l$, we obtain
\begin{align*}
\{\FDP(R_{\hat{\l}},P)>\gamma \}&\subset\{ \l^\star\leq \hat{\l}\}\\
&\subset \{|\cH_0(P)\cap R_{\l^\star}| \geq  \lfloor\gamma\l^\star\rfloor +1\}.
\end{align*}
Moreover, if $\l^\star\geq 2$, by definition of $\l^\star$, we have $\l^\star-1 \notin \mathcal{L}$. Hence, we obtain the following upper-bound for $|\cH_0(P)|$: 
$$|\cH_0(P)| = m-  |\cH_1(P)|\leq m-  |\cH_1(P)\cap R_{\l^\star-1}| \leq m-\l^{\star}+ \lfloor\gamma(\l^\star-1)\rfloor+1.$$
Since the above bound is also true when $\l^\star=1$, it holds for any possible value of $\l^\star$.

Finally noting that $\l^\star$ only depends on the variable set  $\{|\cH_1(P)\cap R_{\l'} |, 1\leq \l' \leq m\}$ and using \eqref{FWC}-\eqref{DA}, we have proved that for any $\l\in\{1,...,m\}$,
\begin{align*}
\P(\FDP(R_{\hat{\l}},P)>\gamma \cond\l^\star=\l)&\leq \P(|\cH_0(P)\cap R_{\l}| \geq  \lfloor\gamma\l\rfloor +1\cond\l^\star=\l)\\
&=\P(|\cH_0(P)\cap R_{\l}| \geq  \lfloor\gamma\l\rfloor +1)\\&\leq \alpha.
\end{align*}
Also, the probability $\P(\FDP(R_{\hat{\l}},P)>\gamma \cond\l^\star=m+1)$ is zero, because it is smaller than $\P(\hat{\l}\in\mathcal{L} \cond\l^\star=m+1)$. This leads to the following result.

\begin{theorem}\label{FDPfromkFWER}
Assume that there exists a family $\{R_\l\}_{1\leq \l \leq m}$ of multiple testing procedures   and a class of distributions $\mathcal{P}'\subset \cP$ satisfying the conditions \eqref{ND:FDP}-\eqref{FWC}-\eqref{DA} defined in Section~\ref{sec:volindex}. Consider the procedure $R_{\hat{\l}}$ where
\begin{align}\label{eq-step-down}
\hat{\l}=\max\big\{ \l\in\{0,...,m\}\telque \forall \l' \in\{0,...,\l\}, \: |R_{\l'}|\geq \l'\big\},
\end{align}
(with the convention $R_0=\emptyset$).
Then $R_{\hat{\l}}$ controls the FDP in the following sense:
\begin{equation}\label{equ-FDPcontrol}
\forall P\in\cP', \:\:\P(\FDP(R_{\hat{\l}},P)>\gamma)\leq \alpha.
\end{equation}
\end{theorem}

The algorithm performed to find \eqref{eq-step-down} is a step-down algorithm; it starts from small rejection numbers and stops the first time that $|R_{\l}|$ is below $\l$. Note that the maximum in \eqref{eq-step-down} is well defined because $\l=0$ satisfies $|R_{\l}|\geq \l$. Furthermore, using \eqref{ND:FDP}, relation \eqref{eq-step-down} implies  $\hat{\l}\leq |R_{\hat{\l}}| \leq |R_{\hat{\l}+1}| < \hat{\l}+1$, so that $|R_{\hat{\l}}|=|R_{\hat{\l}+1}|= \hat{\l}$ holds. As a consequence, the procedure of Theorem~\ref{FDPfromkFWER} can be equivalently defined by $R_{\tilde{\l}}$ where 
\begin{align}\label{eq-step-down2}
\tilde{\l}=\min\{ \l\in\{1,...,m+1\}\telque |R_\l | \leq  \l- 1\},
\end{align}
with the convention $R_{m+1}=R_m$ (so that the minimum in \eqref{eq-step-down2} is well defined).

\subsection{Theorem~3.1 of \cite{LR2005} and the quantile-binomial procedure as corollaries}\label{sec:FDPcor1}

Going back to the specific setting \eqref{DA2}  described in Section~\ref{sec:volindex}, we may derive from Theorem~\ref{FDPfromkFWER} the following corollary.

\begin{corollary}\label{cor-LR2005}
Let us consider the deterministic threshold family $(t^{LR}_\l)_{1\leq \l\leq m}$ defined by \eqref{equ_seuilRom} and consider
\begin{align}
\hat{\l}
&= \max\big\{ \l\in\{0,...,m\}\telque \forall \l' \in\{0,...,\l\}, \:  p_{(\l')}\leq t^{LR}_{\l'} \big\},
\label{equ-step-down}
\end{align}
where $0=p_{(0)}\leq p_{(1)}\leq ...\leq p_{(m)}$ denote the ordered $p$-values and by convention $t^{LR}_0=0$.
Then the procedure $R_{\hat{\l}}=\{i \in\{1,...,m\}\telque p_i\leq t^{LR}_{\hat{\l}}\}$ satisfies the FDP control \eqref{equ-FDPcontrol} for the subset $\mathcal{P}'$ of distributions  $P\in\cP$ such that 
the family $(p_i(X))_{i\in\cH_0(P)}$ is independent of the family $(p_i(X))_{i\in\cH_1(P)}$.
\end{corollary}

By reproducing the end of the proof of Theorem~\ref{FDPfromkFWER}  in the particular setting of Corollary~\ref{cor-LR2005}, we may increase a bit the distribution set $\mathcal{P}'$ in Corollary~\ref{cor-LR2005} to the set of $P\in\cP$ such that 
for any $i\in\cH_0(P)$, $\forall u\in[0,1]$, $\P(p_i(X)\leq u\cond (p_i(X))_{i\in\cH_1(P)})\leq u$. This is the distributional setting of  Theorem~3.1 of \cite{LR2005}.
Hence, we are able to recover the latter result (with a slight improvement in the threshold family).

Furthermore, if we want to ensure the FDP control \eqref{equ-FDPcontrol} only for the smaller distribution set $\cP'=\cP^I$, we may consider the larger threshold family $(t^{Q}_\l)_{1\leq \l\leq m}$ defined by \eqref{equ_TLR_improved}. This gives rise to the 
step-down
 procedure 
\begin{equation}R^{Q}=\{i \in\{1,...,m\}\telque p_i\leq t^Q_{\hat{\l}}\},\label{binomproc}
\end{equation}
 where $\hat{\l}= \max\{ \l\in\{0,...,m\}\telque \forall \l' \in\{0,...,\l\}, \:  p_{(\l')}\leq t^{Q}_{\l'} \}$ (with $t^Q_0=0$). The latter is the procedure described in Algorithm~\ref{algo-quant-binom}, because $p_{(\l)}\leq t^Q_\l$ if and only if $q_{\l}(p_{(\l)})\leq \gamma \l$, with $q_{\l}(\cdot)$ defined by \eqref{equ-quantilebinom}.
As a consequence, Theorem~\ref{FDPfromkFWER} provides the result announced in Section~\ref{sec:present}.

\begin{corollary}\label{cor-LR2005-improved}
For any $\gamma,\alpha\in(0,1)$, the quantile-binomial procedure $R^{Q}$ described in Algorithm~\ref{algo-quant-binom}, or equivalently in \eqref{binomproc}, 
controls the FDP in the following way:
$$
\forall P\in\cP^I, \:\:\P(\FDP(R^{Q},P)>\gamma)\leq \alpha.
$$
In particular, the median-binomial procedure $R^{M}$ (using $\alpha=1/2$) provides that the median of the distribution of $\FDP(R^{M},P)$ is controlled at level $\gamma$ for any $P\in\cP^I$.
\end{corollary}

To our knowledge, the above result is a new finding. It establishes a FDP control  which is substantially more suitable to the case of independent $p$-values in comparison with the procedure of \cite{LR2005}. 
Further comments on this procedure can be found in Section~\ref{sec:binom-BH}.

\subsection{Theorem~4.1 (i) of \cite{RW2007} as a corollary}\label{sec:FDPcor2}

In Section~4 of \cite{RW2007}, a step-down procedure $S_{\hat{k}}$ is defined from a generic family $\{S_{k}\}_{1\leq k \leq m}$ of thresholding based procedures. 
The latter family is assumed to be such that each $S_k$ controls the $k$-FWER for $1\leq k \leq m$ and $S_{k}\subset S_{k+1}$ for $1\leq k \leq m-1$. The index $\hat{k}$ is obtained as follows:
\begin{equation}
\label{equ-RW2007}
\hat{k}= \min\{k\in \{1,...,m+1\} \telque \gamma |S_k| < k - \gamma\},
\end{equation}
where we use here  the convention $S_{m+1}=S_m$ (so that the above set always contains $k=m+1$).
Theorem~4.1 (i) of \cite{RW2007} states that $S_{\hat{k}}$ controls the FDP in the asymptotic sense, as the sample size available to perform each test tends to infinity. This can be seen as a (non-asymptotic) FDP control in a Dirac configuration where the $p$-values corresponding to false nulls are equal to zero.
Set under this form, Theorem~4.1 (i) of \cite{RW2007} can be derived from Theorem~\ref{FDPfromkFWER}. 

For this, let $R_\l=S_{\lfloor\gamma {\l}\rfloor +1 }$, for $\l\in\{1,...,m\}$, and note that the family $\{R_\l\}_{1\leq \l\leq m}$ satisfies \eqref{ND:FDP}-\eqref{FWC} and \eqref{DA3}, by taking the distribution set $\cP'$ corresponding to Dirac configurations for the $p$-values. Hence, Theorem~\ref{FDPfromkFWER} establishes the FDP control for the Dirac configurations of the procedure $R_{\tilde{\l}}$ where $\tilde{\l}$ is defined by \eqref{eq-step-down}, or equivalently by \eqref{eq-step-down2}. Thus, it only remains to show that the  step-down algorithms \eqref{equ-RW2007} and \eqref{eq-step-down2} lead to the same procedure, that is, 
\begin{equation}
\nonumber
R_{\tilde{\l}} = S_{\hat{k}}.
\end{equation}
To prove the latter, we establish $\hat{k}=\lfloor\gamma \tilde{\l}\rfloor +1$. First, using \eqref{eq-step-down2}, $\tilde{\l}$ satisfies $\gamma |S_{\lfloor\gamma \tilde{\l}\rfloor +1 }| \leq  \gamma \tilde{\l}- \gamma$. Since $\gamma \l<\lfloor\gamma {\l}\rfloor +1$, we deduce from the definition of $\hat{k}$ that $\lfloor\gamma \tilde{\l}\rfloor +1 \geq \hat{k}$. Conversely, by considering the unique integer $\l\in\{1,...,m\}$ satisfying $\hat{k}/\gamma-1\leq \l< \hat{k}/\gamma$ and thus also $\lfloor\gamma {\l}\rfloor +1= \hat{k}$, we have that for any integer $j$, $\gamma j<  \hat{k} \Rightarrow j\leq \l$. Applying the latter for $j=|S_{\hat{k}}|+1$, we obtain from $\gamma (|S_{\hat{k}}|+1) < \hat{k}$ that $|S_{\hat{k}}|\leq \l-1$ and thus $\l\geq \tilde{\l}$, by using the definition of $\tilde{\l}$. This in turn implies $\hat{k}\geq \lfloor\gamma {\tilde{\l}}\rfloor+1$.
We thus have proved the following result, which can be seen as Theorem~4.1 (i) of \cite{RW2007} in the Dirac setting.

\begin{corollary}\label{FDPfromkFWER-DU}
Assume that there exists a family $\{S_k\}_{1\leq k \leq m}$ of multiple testing procedures (with the convention $S_{m+1}=S_m$) satisfying 
\begin{itemize}
\item[-] for each $k\in\{1,...,m\}$, $S_k$ is of the  form $\{i\in\{1,...,m\}\telque p_i\leq t_k(\mbf{p})\}$ for a possibly data-dependent threshold $t_k(\cdot)\in[0,1]$;
\item[-] for each $k\in\{1,...,m-1\}$, $S_{k}\subset S_{k+1}$;
\item[-] for each $k\in\{1,...,m\}$, $\forall P\in\cP$, $\kFWER(S_k,P)\leq \alpha$.
\end{itemize}
Consider $\hat{k}$ defined in  \eqref{equ-RW2007}
and the subset $\cP'$ of distributions $P\in\cP$ corresponding to a Dirac configuration, i.e., such that $\forall P\in \cP'$, $\forall i\in\cH_1(P)$, $p_i(x)=0$ for $P$-almost every $x\in\cX$. Then we have $\forall P\in\cP', \:\:\P(\FDP(S_{\hat{k}},P)>\gamma)\leq \alpha.$
\end{corollary}

\section{Discussion}\label{sec:discuss}

\subsection{Complexity of the $\kFWER$ step-down approach}

One major limitation of the $\kFWER$ approach presented in Section~\ref{sec:kFWER} is that the computation of $\phi(\cdot)$ in  \eqref{equ-pfkFWER} can become cumbersome  when $k$ is large because we should consider all subsets $I$ of $\cC^c$ of cardinality $k-1$ (say that $|\cC^c|\geq k-1$). However, 
we may modify this algorithm by considering only the set $I$ equals to the $k-1$ indexes of $\cC^c$ corresponding to the $k-1$ largest $p$-values in $\{p_i,i\in\cC^c\}$. As noted in \cite{RW2007}, this ``streamlined" step-down procedure still controls the $\kFWER$ in the Dirac model where each false null has a $p$-value equals to zero. The latter is true because in this model, 
as soon as $|\cC^c\cap \cH_0(P)| \leq k-1$, we know that the set $\cC^c\cap \cH_0(P)$ is included in the set $I$ of indexes corresponding to the $k-1$ largest $p$-values in $\{p_i,i\in\cC^c\}$ (because the $p$-values of $\{p_i,i\in\cC^c\cap \cH_1(P)\}$ are zero). Nevertheless, no proof of this $\kFWER$ control stands without this Dirac assumption.

\subsection{FDR control is not FDP control}\label{sec:discussFDP}

Since the only interpretable variable is the FDP and not its expectation, controlling the FDR is
 meaningful only when the FDP concentrates well around the FDR. As the hypothesis number $m$ grows, Neuvial (2008) showed that the latter holds for step-up type procedures when a Donsker type theorem for the e.c.d.f. is valid, so for instance under independence or ``weak" dependence, see \cite{Neu2008}. However, under some unspecified dependencies, we do not know how the FDP concentrates. For instance, even under a very simple $\rho$-equi-correlated Gaussian model (corresponding to Example~\ref{sec:vectorpos}, where the non-diagonal entries of $\Sigma(P)$ are all equal to $\rho$), its was shown in \cite{DR2011} that the convergence rate of the FDP to the FDR can be arbitrarily slow when $\rho=\rho_m$ tends to zero as $m$ tends to infinity. Additionally, it was proved in \cite{FDR2007} that no concentration phenomenon occurs when $\rho$ is kept fixed with $m$.
Also, as shown in \cite{RV2010}, the ``sparsity" ($\pi_0(P)=\pi_{0,m}(P)$ tends to $1$ as $m$ tends to infinity) is one other feature that can slow down the FDP convergence. 
Therefore, 
in all these cases, the FDP convergence is slow and controlling the FDR does not lead to a clear interpretation for the underlying FDP. 
The latter drawback does not arise while controlling the FDP upper-tail distribution: 
for instance, the FDP control $\P(\FDP> 0.01)\leq 0.5$ ensures that, with a probability at least 0.5, the FDP is below $0.01$, and this interpretation holds whatever the FDP distribution is.
However, the FDR stays useful, because 
this is a simpler criterion for which the controlling methodology is (for now) much more developed in comparison with the FDP controlling methodology.

\subsection{Quantile-binomial procedure and relation to previous work}\label{sec:binom-BH}

Let us consider the quantile-binomial procedure defined in algorithm~\ref{algo-quant-binom} and the quantile function $q_{\l}(\cdot)$ defined by \eqref{equ-quantilebinom}. In the particular case where we take $\alpha=1/2$, the procedure is called the median-binomial procedure and Corollary~\ref{cor-LR2005-improved} shows that it controls the median of the FDP at level $\gamma$ under independence of the $p$-values. Interestingly, in the ``Gaussian regime" where the underlying binomial variable is close to a Gaussian variable (say, $\gamma$ not too small, many rejections), the median is close to the expectation and thus $q_\l(t)\simeq (m -\l+ \lfloor\gamma ({\l}-1)\rfloor + 1) t \simeq (m-(1-\gamma)\l+1) t$. Hence, in this case, the median-binomial procedure is close to the step-down procedure using the thresholding $t_\l=\gamma \l/(m-(1-\gamma)\l+1)$. As matter of fact, the latter procedure has been recently introduced by Gavrilov et al. (2009) and it has been proved to control the FDR under independence, see \cite{GBS2009}. Roughly speaking, the latter may be interpreted in our framework as a ``mean-binomial procedure". However, in the Poisson regime (say, $\gamma$ small, few rejections), the median-binomial procedure can be substantially different from the procedure of Gavrilov et al. (2009). Hence, we should keep in mind that the two procedures do not control the same error rate. 
These different remarks are illustrated in Figure~\ref{fig_comparthreshold}, where we have also reported the Benjamini-Hochberg threshold.

 \begin{figure}[h!]
\begin{center}
$\gamma=0.01$\\\vspace{-0.7cm}
\includegraphics[scale=0.25,angle=-90]{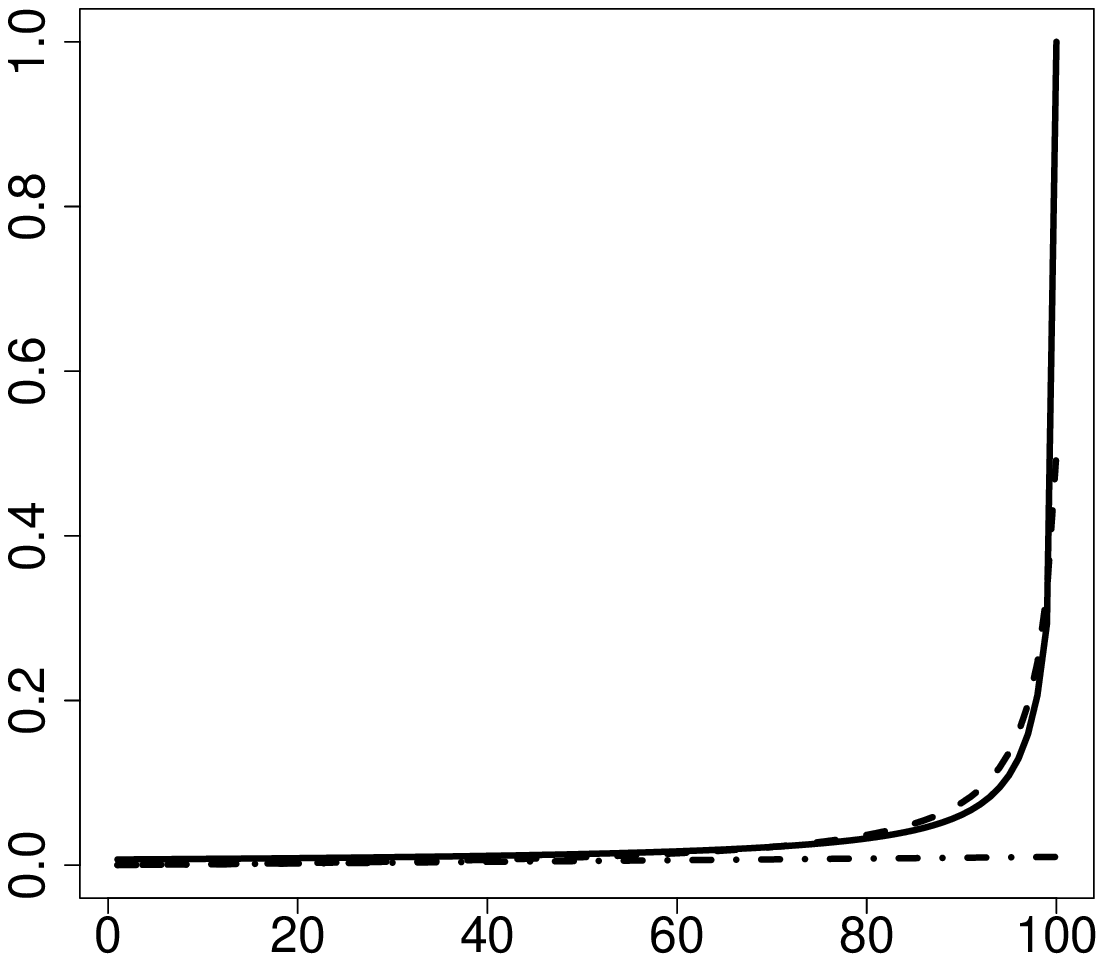}
\includegraphics[scale=0.25,angle=-90]{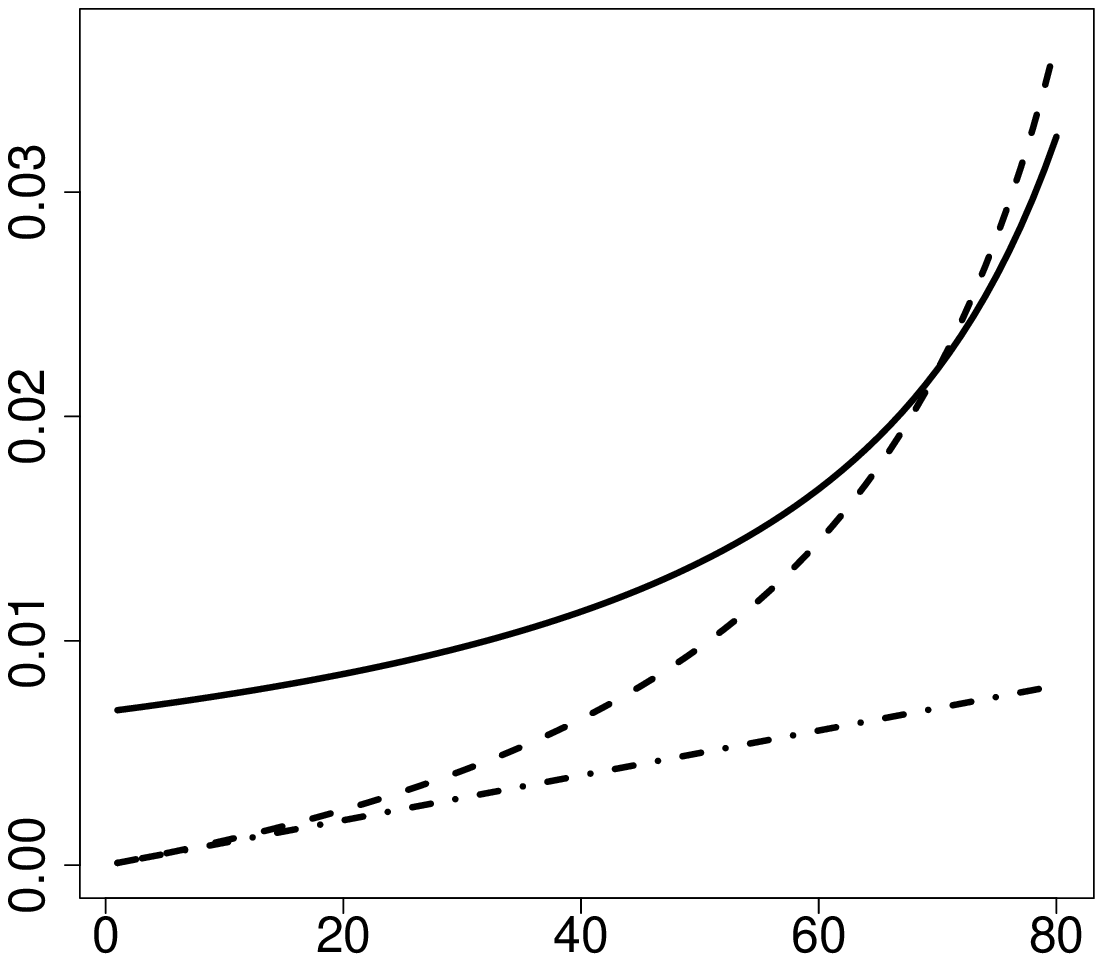}\\
$\gamma=0.1$\\\vspace{-0.7cm}
\includegraphics[scale=0.25,angle=-90]{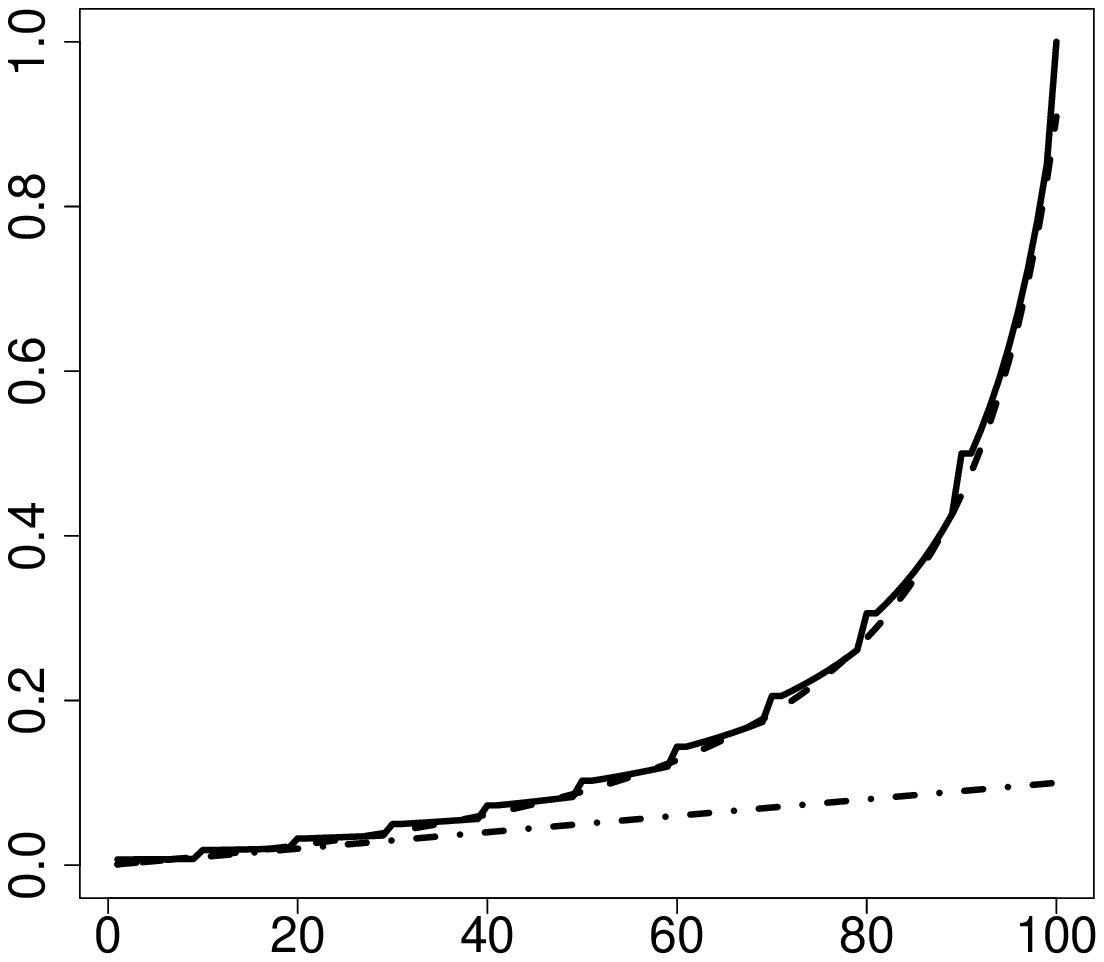}
\includegraphics[scale=0.25,angle=-90]{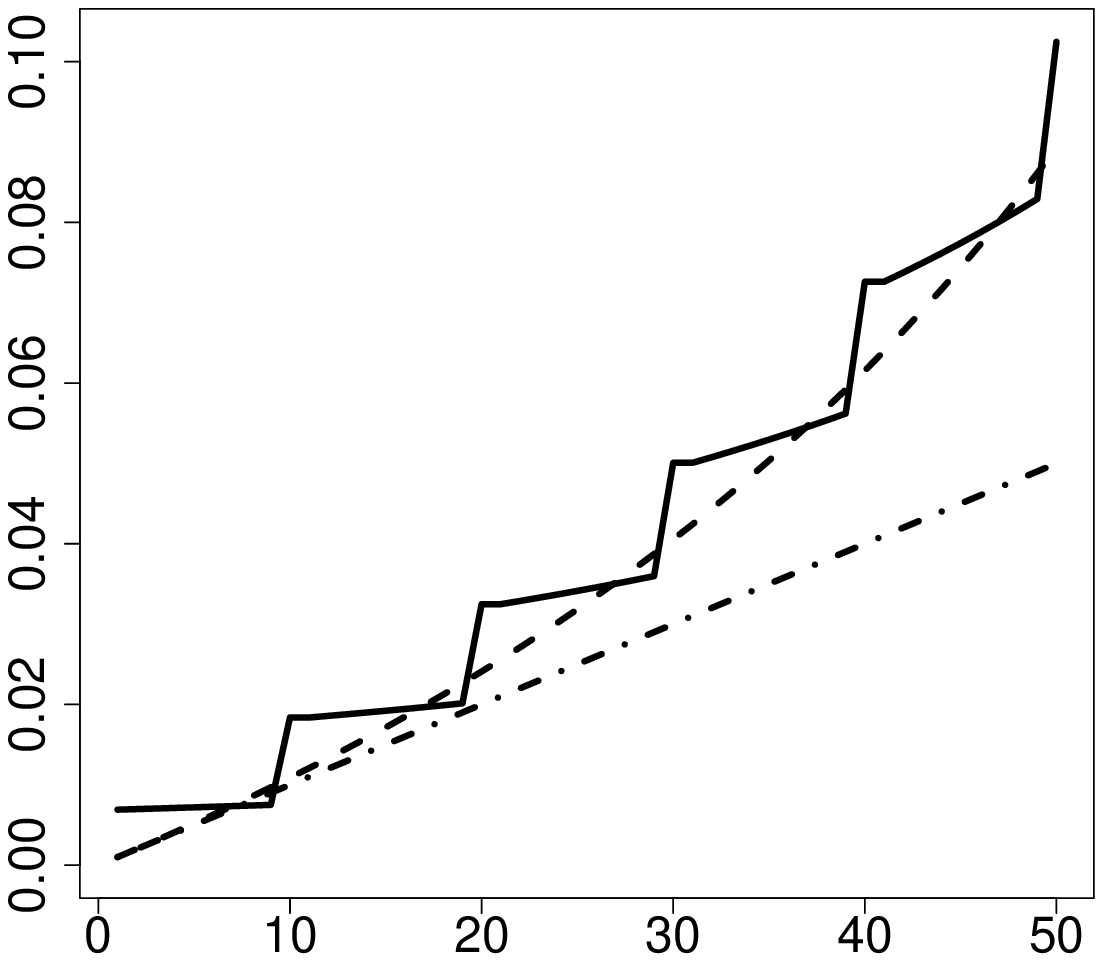}\\
\caption{Comparison between Benjamini-Hochberg thresholding $t_\l=\gamma \l/m$ (dashed-dotted), the Gavrilov et al. thresholding $t_\l=\gamma \l/(m-(1-\gamma)\l+1)$ (dashed) and the quantile thresholding  $t^Q_{\l}$ defined by \eqref{equ_TLR_improved} with $\alpha=0.5$ (solid) in function of $\l$. 
 $m=100$; Top: $\gamma=0.01$; Bottom: $\gamma=0.1$. Each right picture is a zoom of the left picture into the region $\l\in\{1,...,80\}$ (top) or $\l\in\{1,...,50\}$ (bottom).}
\label{fig_comparthreshold}
\end{center}
\end{figure}

\subsection{Conclusion}\label{sec:future}

In this paper, we have recovered some of the classical state-of-the-art multiple testing procedures for controlling the FDR, $\kFWER$ and the FDP. 
Additionally, some new contributions were also given for $\kFWER$ and FDP control, by extending and unifying some previous work of multiple testing literature and by finding a novel procedure, based on the quantiles of the binomial distribution, which controls the FDP under independence.

The type I error rate control research area still has many unsolved issues. Among the major  concerns, the FDP control in Section~\ref{sec:FDP} needs a very strong distributional assumption
on the test statistics, namely independence or ``Dirac" assumption. To our knowledge, no procedure adaptive to dependencies is proved to control the FDP without assuming such a strong requirement.
This is a room left for future developments, which would have a strong impact on high-dimensional data analysis. 

 \section*{Acknowledgements}

I am very grateful to Lucien Birg\'e and Gilles Blanchard for their helpful and particularly relevant comments. 
I also warmly thank Tabea Rebafka and Sylvain Arlot  for having carefully read a previous version of this manuscript. 
This work was supported by the French Agence Nationale de la Recherche (ANR grant references: ANR-09-JCJC-0027-01, ANR-PARCIMONIE, ANR-09-JCJC-0101-01) and the French ministry of foreign and european affairs (EGIDE - PROCOPE project number 21887 NJ).

\appendix

\section{Defining a $p$-value from a test statistic}\label{sec:statpvalue}

Let us consider the problem of testing a (single) hypothesis $H_0:$ ``$P \in \Theta_0$"  from a test statistic $S(X)$. Assume that $H_0$ should be rejected for ``large" values of $S(X)$. We let $T_P(s)=\P_{X\sim P}(S(X)\geq s)$, $F_P(s)=\P_{X\sim P}(S(X)\leq s)$ and  $F^{-1}_P(v)=\min\{s\in\mathbb{R}\cup \{-\infty\}\telque F_P(s)\geq v\}$. 
The following  result is elementary and can be considered as well known. It is strongly related to Theorem~10.12 in \cite{Was2004}, Lemma~3.3.1 in \cite{LR2005b} (see also Problem 3.23 therein) and Proposition~1.2 in \cite{DL2008}.

\begin{proposition}\label{prop:statpvalue}
The $p$-value $p(X)=\sup_{P\in \Theta_0} T_P(S(X))$ satisfies the following:
\begin{itemize}
\item[(i)] $p(X)$ is stochastically lower-bounded by a uniform variable under the null, that is, $$\forall P \in \Theta_0,\:\:\forall u\in[0,1], \: \:\:\:\P_{X\sim P}(p(X)\leq u)\leq u.$$
\item[(ii)] if for any $P\in\Theta_0$, $F_P$ is continuous, we have for any realization $x$ of $X$,
\begin{align*}
p(x)&=\min\{\alpha\in[0,1]\telque S(x)\geq \sup_{P\in\Theta_0} F^{-1}_P(1-\alpha)\}
.
\end{align*}
If additionally $\Theta_0$ is a singleton, $p(X)\sim U(0,1)$ whenever $P\in\Theta_0$. 
\item[(iii)] 
if for any $P\in\Theta_0$, the variable $S(X)$ takes its values in a discrete set with probability $1$, we have for any realization $x$ of $X$,
$$p(x)=\min\{\alpha\in[0,1]\telque S(x)> \sup_{P\in\Theta_0} F^{-1}_P(1-\alpha)\}.$$
In particular, if $S(X)$ is an integer random variable,  we have for any $x$ such that $S(x)\in\mathbb{N}$,
$$p(x)=\min\{\alpha\in[0,1]\telque S(x)\geq  \sup_{P\in\Theta_0} F^{-1}_P(1-\alpha)+1\}.$$
\end{itemize}
\end{proposition}

A consequence is that the two classical definitions of a $p$-value are compatible in the following way.
\begin{corollary}
Assume that there exists $Q\in \Theta_0$ such that for any $P\in\Theta_0$, for all $s\in\R$, $F_P(s)\geq F_Q(s)$.
Let $p(X)= T_Q(S(X))$ and consider the families of tests $\{\phi_\alpha\}_{\alpha\in[0,1]}$ and $\{\phi'_\alpha\}_{\alpha\in[0,1]}$, where $\phi_\alpha(x)=\ind{S(x)\geq F^{-1}_Q(1-\alpha)}$ and $\phi'_\alpha(x)=\ind{S(x)> F^{-1}_Q(1-\alpha)}$. Then the following holds. \begin{itemize}
\item[(i)] 
if $F_Q$ is continuous, the tests $\phi_\alpha$ and $\phi'_\alpha$ are of level $\alpha$ for all $\alpha\in[0,1]$ and
we have for any realization $x$ of $X$,
\begin{align*}
[p(x), 1]&=\{\alpha\in[0,1]\telque \phi_\alpha(x)=1\}.
\end{align*}
and for $Q$-almost every $x$,
\begin{align*}
(p(x), 1]&=\{\alpha\in[0,1]\telque \phi'_\alpha(x)=1\}.
\end{align*}
\item[(ii)] 
if for $X\sim Q$ the variable $S(X)$ takes its values in a discrete set with probability $1$, the test $\phi'_\alpha$ is of level $\alpha$ while the test $\phi_\alpha$ is not of level $\alpha$, for all $\alpha\in[0,1]$, 
and we have for any realization $x$ of $X$,
\begin{align*}
[p(x), 1]&=\{\alpha\in[0,1]\telque \phi'_\alpha(x)=1\}.
\end{align*}
\end{itemize}
In particular,  we have both in the continuous and discrete case that for $Q$-almost every $x$,
\begin{align*}
p(x)&=\inf\{\alpha\in[0,1]\telque  \phi'_\alpha(x)=1\}.
\end{align*}
\end{corollary}


\begin{proof}
From Proposition~\ref{prop:statpvalue} (ii) and (iii), the only assertion to be proved is that for all $\alpha\in[0,1]$, for $Q$-almost every $x$,
$S(x)> F_Q^{-1}(1-\alpha) \Leftrightarrow p(x) <\alpha$. Let us denote $\mathcal{Q}=\{F_Q^{-1}(1-\alpha),\alpha\in[0,1]\}$. 
Since ${F}_Q$ is increasing on $\mathcal{Q}$,  the desired relation is provided for $S(x)\in\mathcal{Q}$. 
We can conclude because $\P_{X\sim Q}(S(X)\in\mathcal{Q})=1$. 
\end{proof}

\begin{example}
To illustrate (i) and (iii) of Proposition~\ref{prop:statpvalue},
let us consider the following simple discrete testing setting (coming from Example~3.3.2 in \cite{LR2005b}).
Let $H_0:$ ``$P=P_0$"  where $P_0$ is the uniform distribution on $\{1,...,10\}$ and consider the test statistic $S(X)=X$. We easily see that the $p$-value $T_{P_0}(X)$ is $p(X)=(11-X)/10$. It satisfies $\P(p(X)\leq u)\leq u$, with equality iff $u$ can be written under the form $i/10$ for some integer $i$, $1\leq i \leq 10$. 
Furthermore, rejecting $H_0$ for $p(X)\leq \alpha$ is equivalent to reject $H_0$ whenever $X\geq k(\alpha)$ where  $k(\alpha)$ is the unique integer satisfying $(11- k(\alpha))/10 \leq \alpha < (12- k(\alpha))/10$. We merely check that $k(\alpha)=F^{-1}_{P_0}(1-\alpha)+1$.
\end{example}

Finally, we provide a proof for Proposition~\ref{prop:statpvalue}.

\begin{proof}
Let $\mathring{F}_P(s)=\P_{X\sim P}(S(X)< s)$ and let us first state the following result: for any $P$, for any $\alpha\in[0,1]$,
\begin{align}
\{T_P(S(X))\leq \alpha\} = \left\{\begin{array}{ll}
\{S(X)\geq F_P^{-1}(1-\alpha)\} &\mbox{ if } \mathring{F}_P(F_P^{-1}(1-\alpha))=1-\alpha\\
\{S(X)>F_P^{-1}(1-\alpha)\} &\mbox{ otherwise } 
\end{array}\right..
\label{rel-fond-pvalue}
\end{align}
To establish \eqref{rel-fond-pvalue}, first note that $\{T_P(S(X))\leq \alpha\}=\{\mathring{F}_P(S(X))\geq 1-\alpha\}\subset \{S(X)\geq F_P^{-1}(1-\alpha)\}$, by definition of $F_P^{-1}(1-\alpha)$. On the one hand, if  $\mathring{F}_P(F_P^{-1}(1-\alpha))=1-\alpha$, we have $\{S(X)\geq F_P^{-1}(1-\alpha)\}\subset \{\mathring{F}_P(S(X))\geq \mathring{F}_P(F_P^{-1}(1-\alpha))\}=\{\mathring{F}_P(S(X))\geq 1-\alpha\}$. On the other hand, if $\mathring{F}_P(F_P^{-1}(1-\alpha))<1-\alpha$, we have $\{\mathring{F}_P(S(X))\geq 1-\alpha\}\subset \{S(X)> F_P^{-1}(1-\alpha)\}$ and $\{S(X)> F_P^{-1}(1-\alpha)\}\subset \{\mathring{F}_P(S(X))\geq  F_P(F_P^{-1}(1-\alpha))\}\subset \{\mathring{F}_P(S(X))\geq  1-\alpha\}$. This proves \eqref{rel-fond-pvalue}.

Let us now prove (i). We have for any $P\in \Theta_0$,
$
\P_{X\sim P}(p(X)\leq \alpha) \leq \P_{X\sim P}(T_P(S(X))\leq \alpha). 
$
Next, applying \eqref{rel-fond-pvalue}, we have if $\mathring{F}_P(F_P^{-1}(1-\alpha))=1-\alpha$, 
\begin{align*}
\P_{X\sim P}(p(X)\leq \alpha) \leq \P_{X\sim P}(S(X)\geq F_P^{-1}(1-\alpha)) = 1-  \mathring{F}_P(F_P^{-1}(1-\alpha)) = \alpha
\end{align*}
and if  $\mathring{F}_P(F_P^{-1}(1-\alpha))<1-\alpha$, 
\begin{align*}
\P_{X\sim P}(p(X)\leq \alpha) \leq \P_{X\sim P}(S(X)> F_P^{-1}(1-\alpha)) = 1-  {F}_P(F_P^{-1}(1-\alpha)) \leq \alpha.
\end{align*}

Assume now that for any $P\in\Theta_0$, $F_P$ is continuous, and prove (ii). In this case, $\mathring{F}_P(F_P^{-1}(1-\alpha))={F}_P(F_P^{-1}(1-\alpha))=1-\alpha$ for any $\alpha\in [0,1]$, so that \eqref{rel-fond-pvalue} provides that $\{T_P(S(X))\leq \alpha\} =\{S(X)\geq F_P^{-1}(1-\alpha)\}$. Hence, we obtain for any realization $x$ of $X$,
\begin{align*}
p(x)&=\min\{\alpha\in[0,1]\telque \forall P\in\Theta_0, T_P(S(x))\leq \alpha\}\\
&=\min\{\alpha\in[0,1]\telque \forall P\in\Theta_0, S(x)\geq F^{-1}_P(1-\alpha)\},
\end{align*}
which leads to the desired result.

For (iii), the proof is similar by noting that $\mathring{F}_P(F_P^{-1}(1-\alpha))<1-\alpha$ in the case where the distribution of $S(X)$ has a discrete support under the null.
\end{proof}
\bibliography{Roquain-jsfds-version2}

\begin{thebibliography}{10}

\bibitem{ABR2010a}
S.~Arlot, G.~Blanchard, and E.~Roquain.
\newblock Some nonasymptotic results on resampling in high dimension. {I}.
  {C}onfidence regions.
\newblock {\em Ann. Statist.}, 38(1):51--82, 2010.

\bibitem{ABR2010b}
S.~Arlot, G.~Blanchard, and E.~Roquain.
\newblock Some nonasymptotic results on resampling in high dimension. {II}.
  {M}ultiple tests.
\newblock {\em Ann. Statist.}, 38(1):83--99, 2010.

\bibitem{BHL2003}
Y.~Baraud, S.~Huet, and B.~Laurent.
\newblock Adaptive tests of linear hypotheses by model selection.
\newblock {\em Ann. Statist.}, 31(1):225--251, 2003.

\bibitem{BH2007}
Y.~Benjamini and R.~Heller.
\newblock False discovery rates for spatial signals.
\newblock {\em J. Amer. Statist. Assoc.}, 102(480):1272--1281, 2007.

\bibitem{BH1995}
Y.~Benjamini and Y.~Hochberg.
\newblock Controlling the false discovery rate: a practical and powerful
  approach to multiple testing.
\newblock {\em J. Roy. Statist. Soc. Ser. B}, 57(1):289--300, 1995.

\bibitem{BH2000}
Y.~Benjamini and Y.~Hochberg.
\newblock On the adaptive control of the false discovery rate in multiple
  testing with independent statistics.
\newblock {\em J. Behav. Educ. Statist.}, 25:60--83, 2000.

\bibitem{BKY2006}
Y.~Benjamini, A.~M. Krieger, and D.~Yekutieli.
\newblock Adaptive linear step-up procedures that control the false discovery
  rate.
\newblock {\em Biometrika}, 93(3):491--507, 2006.

\bibitem{BY2001}
Y.~Benjamini and D.~Yekutieli.
\newblock The control of the false discovery rate in multiple testing under
  dependency.
\newblock {\em Ann. Statist.}, 29(4):1165--1188, 2001.

\bibitem{Black2004}
M.~A. Black.
\newblock A note on the adaptive control of false discovery rates.
\newblock {\em J. R. Stat. Soc. Ser. B Stat. Methodol.}, 66(2):297--304, 2004.

\bibitem{Blanch2009}
G.~Blanchard.
\newblock Contr{\^o}le non-asymptotique adaptatif du family-wise error rate en
  tests multiples.
\newblock Talk at Journ{\'e}es Statistiques du Sud, Porquerolles, 2009.

\bibitem{BF2007}
G.~Blanchard and F.~Fleuret.
\newblock Occam's hammer.
\newblock In {\em Learning theory}, volume 4539 of {\em Lecture Notes in
  Comput. Sci.}, pages 112--126. Springer, Berlin, 2007.

\bibitem{BR2008}
G.~Blanchard and E.~Roquain.
\newblock Two simple sufficient conditions for {FDR} control.
\newblock {\em Electron. J. Stat.}, 2:963--992, 2008.

\bibitem{BR2009}
G.~Blanchard and E.~Roquain.
\newblock Adaptive false discovery rate control under independence and
  dependence.
\newblock {\em J. Mach. Learn. Res.}, 10:2837--2871, 2009.

\bibitem{CR2010}
A.~Celisse and S.~Robin.
\newblock A cross-validation based estimation of the proportion of true null
  hypotheses.
\newblock {\em Journal of Statistical Planning and Inference}, 140(11):3132 --
  3147, 2010.

\bibitem{CT2008}
Z.~Chi and Z.~Tan.
\newblock Positive false discovery proportions: intrinsic bounds and adaptive
  control.
\newblock {\em Statist. Sinica}, 18(3):837--860, 2008.

\bibitem{DR2011}
S.~Delattre and E.~Roquain.
\newblock On the false discovery proportion convergence under gaussian
  equi-correlation.
\newblock {\em Statistics \& Probability Letters}, 81(1):111--115, 2011.

\bibitem{DL2008}
S.~Dudoit and M.~J. van~der Laan.
\newblock {\em Multiple testing procedures with applications to genomics}.
\newblock Springer Series in Statistics. Springer, New York, 2008.

\bibitem{DR2006}
C.~Durot and Y.~Rozenholc.
\newblock An adaptive test for zero mean.
\newblock {\em Math. Methods Statist.}, 15(1):26--60, 2006.

\bibitem{Efron2008}
B.~Efron.
\newblock Microarrays, empirical {B}ayes and the two-groups model.
\newblock {\em Statist. Sci.}, 23(1):1--22, 2008.

\bibitem{Efron2009}
B.~Efron.
\newblock Correlated z -values and the accuracy of large-scale statistical
  estimates.
\newblock {\em J. Amer. Statist. Assoc.}, 105(491):1042--1055, 2010.

\bibitem{Far2006}
A.~Farcomeni.
\newblock Some results on the control of the false discovery rate under
  dependence.
\newblock {\em Scand. J. Statist.}, 34(2):275--297, 2007.

\bibitem{FZ2006}
J.~A. Ferreira and A.~H. Zwinderman.
\newblock On the {B}enjamini-{H}ochberg method.
\newblock {\em Ann. Statist.}, 34(4):1827--1849, 2006.

\bibitem{FDR2009}
H.~Finner, R.~Dickhaus, and M.~Roters.
\newblock On the false discovery rate and an asymptotically optimal rejection
  curve.
\newblock {\em Ann. Statist.}, 37(2):596--618, 2009.

\bibitem{FDR2007}
H.~Finner, T.~Dickhaus, and M.~Roters.
\newblock Dependency and false discovery rate: asymptotics.
\newblock {\em Ann. Statist.}, 35(4):1432--1455, 2007.

\bibitem{Fish1935}
R.~A. Fisher.
\newblock {\em The Design of Experiments.}
\newblock Oliver and Boyd, Edinburgh.p, 1935.

\bibitem{GBS2009}
Y.~Gavrilov, Y.~Benjamini, and S.~K. Sarkar.
\newblock An adaptive step-down procedure with proven {FDR} control under
  independence.
\newblock {\em Ann. Statist.}, 37(2):619--629, 2009.

\bibitem{GW2002}
C.~Genovese and L.~Wasserman.
\newblock Operating characteristics and extensions of the false discovery rate
  procedure.
\newblock {\em J. R. Stat. Soc. Ser. B Stat. Methodol.}, 64(3):499--517, 2002.

\bibitem{GW2004}
C.~Genovese and L.~Wasserman.
\newblock A stochastic process approach to false discovery control.
\newblock {\em Ann. Statist.}, 32(3):1035--1061, 2004.

\bibitem{GF2010}
J.~Goeman and L.~Finos.
\newblock The inheritance procedure: multiple testing of tree-structured
  hypotheses.
\newblock Technical report, Leiden University medical center, 2010.

\bibitem{GS2010}
J.~Goeman and A.~Solari.
\newblock The sequential rejection principle of familywise error control.
\newblock {\em Ann. Statist.}, 38(6):3782--3810, 2010.

\bibitem{GR2008}
W.~Guo and M.~B. Rao.
\newblock On control of the false discovery rate under no assumption of
  dependency.
\newblock {\em Journal of Statistical Planning and Inference},
  138(10):3176--3188, 2008.

\bibitem{HT1987}
Y.~Hochberg and A.~C. Tamhane.
\newblock {\em Multiple comparison procedures}.
\newblock Wiley Series in Probability and Mathematical Statistics: Applied
  Probability and Statistics. John Wiley \& Sons Inc., New York, 1987.

\bibitem{Holm1979}
S.~Holm.
\newblock A simple sequentially rejective multiple test procedure.
\newblock {\em Scand. J. Statist.}, 6(2):65--70, 1979.

\bibitem{KRW2010}
K.~I. Kim, E.~Roquain, and M.~A. van~de Wiel.
\newblock Spatial clustering of array {CGH} features in combination with
  hierarchical multiple testing.
\newblock {\em Stat. Appl. Genet. Mol. Biol.}, 9(1):Art. 40, 2010.

\bibitem{KW2008}
K.~I. Kim and M.~van~de Wiel.
\newblock Effects of dependence in high-dimensional multiple testing problems.
\newblock {\em BMC Bioinformatics}, 9(1):114, 2008.

\bibitem{LR2005}
E.~L. Lehmann and J.~P. Romano.
\newblock Generalizations of the familywise error rate.
\newblock {\em Ann. Statist.}, 33:1138--1154, 2005.

\bibitem{LR2005b}
E.~L. Lehmann and J.~P. Romano.
\newblock {\em Testing statistical hypotheses}.
\newblock Springer Texts in Statistics. Springer, New York, third edition,
  2005.

\bibitem{Mass2007}
P.~Massart.
\newblock {\em Concentration inequalities and model selection}, volume 1896 of
  {\em Lecture Notes in Mathematics}.
\newblock Springer, Berlin, 2007.
\newblock Lectures from the 33rd Summer School on Probability Theory held in
  Saint-Flour, July 6--23, 2003, With a foreword by Jean Picard.

\bibitem{MMB2009}
N.~Meinshausen, L.~Meier, and P.~B\"{u}hlmann.
\newblock {p-values for high-dimensional regression}.
\newblock {\em J. Amer. Statist. Assoc.}, 104(488):1671--1681, 2009.

\bibitem{Astro2001}
C.~J. Miller, C.~Genovese, R.~C. Nichol, L.~Wasserman, A.~Connolly,
  D.~Reichart, A.~Hopkins, J.~Schneider, and A.~Moore.
\newblock Controlling the false-discovery rate in astrophysical data analysis.
\newblock {\em The Astronomical Journal}, 122(6):3492--3505, 2001.

\bibitem{Neu2008}
P.~Neuvial.
\newblock Asymptotic properties of false discovery rate controlling procedures
  under independence.
\newblock {\em Electron. J. Stat.}, 2:1065--1110, 2008.

\bibitem{PNBL2005}
D.~Pantazis, T.~E. Nichols, S.~Baillet, and R.~M. Leahy.
\newblock A comparison of random field theory and permutation methods for
  statistical analysis of meg data.
\newblock {\em NeuroImage}, 25:383--394, 2005.

\bibitem{RS2006b}
J.~P. Romano and A.~M. Shaikh.
\newblock Stepup procedures for control of generalizations of the familywise
  error rate.
\newblock {\em Ann. Statist.}, 34(4):1850--1873, 2006.

\bibitem{RW2005}
J.~P. Romano and M.~Wolf.
\newblock Exact and approximate stepdown methods for multiple hypothesis
  testing.
\newblock {\em J. Amer. Statist. Assoc.}, 100(469):94--108, 2005.

\bibitem{RW2007}
J.~P. Romano and M.~Wolf.
\newblock Control of generalized error rates in multiple testing.
\newblock {\em Ann. Statist.}, 35(4):1378--1408, 2007.

\bibitem{RW2010}
J.~P. Romano and M.~Wolf.
\newblock Balanced control of generalized error rates.
\newblock {\em Ann. Statist.}, 38(1):598--633, 2010.

\bibitem{RW2009}
E.~Roquain and M.~van~de Wiel.
\newblock Optimal weighting for false discovery rate control.
\newblock {\em Electron. J. Stat.}, 3:678--711, 2009.

\bibitem{RV2010}
E.~Roquain and F.~Villers.
\newblock Exact calculations for false discovery proportion with application to
  least favorable configurations.
\newblock {\em Ann. Statist.}, 39(1):584--612, 2011.

\bibitem{RDV2006}
D.~Rubin, S.~Dudoit, and M.~van~der Laan.
\newblock A method to increase the power of multiple testing procedures through
  sample splitting.
\newblock {\em Stat. Appl. Genet. Mol. Biol.}, 5:Art. 19, 20 pp. (electronic),
  2006.

\bibitem{Sar2002}
S.~K. Sarkar.
\newblock Some results on false discovery rate in stepwise multiple testing
  procedures.
\newblock {\em Ann. Statist.}, 30(1):239--257, 2002.

\bibitem{Sar2008}
S.~K. Sarkar.
\newblock On methods controlling the false discovery rate.
\newblock {\em Sankhya, Ser. A}, 70:135--168, 2008.

\bibitem{SS1982}
T.~Schweder and E.~Spj{\o}tvoll.
\newblock {Plots of P-values to evaluate many tests simultaneously}.
\newblock {\em Biometrika}, 69(3):493--502, 1982.

\bibitem{See1968}
P.~Seeger.
\newblock A note on a method for the analysis of significances en masse.
\newblock {\em Technometrics}, 10(3):586--593, 1968.

\bibitem{Spo1996}
V.~G. Spokoiny.
\newblock Adaptive hypothesis testing using wavelets.
\newblock {\em Ann. Statist.}, 24(6):2477--2498, 1996.

\bibitem{Storey2002}
J.~D. Storey.
\newblock A direct approach to false discovery rates.
\newblock {\em J. R. Stat. Soc. Ser. B Stat. Methodol.}, 64(3):479--498, 2002.

\bibitem{STS2004}
J.~D. Storey, J.~E. Taylor, and D.~Siegmund.
\newblock Strong control, conservative point estimation and simultaneous
  conservative consistency of false discovery rates: a unified approach.
\newblock {\em J. R. Stat. Soc. Ser. B Stat. Methodol.}, 66(1):187--205, 2004.

\bibitem{WBW2009}
M.~A. van~de Wiel, J.~Berkhof, and W.~N. van Wieringen.
\newblock {Testing the prediction error difference between 2 predictors}.
\newblock {\em Biostat}, 10(3):550--560, July 2009.

\bibitem{WW2007}
M.~A. van~de Wiel and W.~N. van Wieringen.
\newblock {{C}{G}{H}regions: {D}imension {R}eduction for {A}rray {C}{G}{H}
  {D}ata with {M}inimal {I}nformation {L}oss}.
\newblock {\em Cancer Inform}, 3:55--63, 2007.

\bibitem{LBH2005}
M.~J. van~der Laan, M.~D. Birkner, and A.~E. Hubbard.
\newblock Empirical {B}ayes and resampling based multiple testing procedure
  controlling tail probability of the proportion of false positives.
\newblock {\em Stat. Appl. Genet. Mol. Biol.}, 4:Art. 29, 32 pp. (electronic),
  2005.

\bibitem{VV2010}
N.~Verzelen and F.~Villers.
\newblock Goodness-of-fit tests for high-dimensional {G}aussian linear models.
\newblock {\em Ann. Statist.}, 38(2):704--752, 2010.

\bibitem{Was2004}
L.~Wasserman.
\newblock {\em All of statistics}.
\newblock Springer Texts in Statistics. Springer-Verlag, New York, 2004.
\newblock A concise course in statistical inference.

\bibitem{WR2006}
L.~Wasserman and K.~Roeder.
\newblock Weighted hypothesis testing.
\newblock Technical report, Dept. of statistics, Carnegie Mellon University,
  2006.

\bibitem{WY1993}
P.~H. Westfall and S.~S. Young.
\newblock {\em Resampling-Based Multiple Testing}.
\newblock Wiley, 1993.
\newblock Examples and Methods for $P$- Value Adjustment.

\end{thebibliography}

\end{document}